\documentclass
[preprint,byrevtex,nofootinbib,10pt,balancelastpage,titlepage,bibnotes,twocolumn,pre,secnumarabic,noshowkeys]{revtex4}%
\usepackage{amsfonts}
\usepackage{amsmath}
\usepackage{amssymb}
\usepackage{graphicx}%
\providecommand{\U}[1]{\protect\rule{.1in}{.1in}}
\newtheorem{theorem}{Theorem}

\newtheorem{corollary}[theorem]{Corollary}

\newtheorem{remark}[theorem]{Remark}

\newenvironment{proof}[1][Proof]{\noindent\textbf{#1.} }{\ \rule{0.5em}{0.5em}}
\ifx\pdfoutput\relax\let\pdfoutput=\undefined\fi
\newcount\msipdfoutput
\ifx\pdfoutput\undefined\else
\ifcase\pdfoutput\else
\ifx\paperwidth\undefined\else
\ifdim\paperheight=0pt\relax\else\pdfpageheight\paperheight\fi
\ifdim\paperwidth=0pt\relax\else\pdfpagewidth\paperwidth\fi
\fi\fi\fi
\begin{document}
\preprint{UATP/2002}
\title{A Novel Trick to Overcome the Phase Space Volume Change and the Use of
Hamiltonian Trajectories with an emphasis on the Free Expansion}
\author{P.D. Gujrati}
\email{pdg@uakron.edu}
\affiliation{Department of Physics, Department of Polymer Science, The University of Akron,
Akron, OH 44325}

\begin{abstract}
We extend and successfully apply a recently proposed microstate nonequilibrium
thermodynamics ($\mu$NEQT) to study expansion/contraction processes. Here, the
numbers of initial and final microstates $\left\{  \mathfrak{m}_{k}\right\}  $
are different so they cannot be connected by unique Hamiltonian trajectories.
This commonly happens when the phase space volume changes, and has not been
studied so far using Hamiltonian trajectories that can be inverted to yield an
identity mapping $\mathcal{T}:\mathfrak{m}_{k}(\mathbf{Z}_{\text{in}}%
^{E})\rightleftarrows\mathfrak{m}_{k}(\mathbf{Z}_{\text{fin}}^{E})$ as the
parameter $\mathbf{Z}^{E}$ in the Hamiltonian is changed. We propose a trick
to overcome this hurdle with a focus on free expansion ($P_{\text{vacuum}}=0$)
in an isolated system, where the concept of dissipated work is not clear. The
trick is shown to be thermodynamically consistent and can be extremely useful
in simulation. We justify that it is the thermodynamic average $\Delta
_{\text{i}}W\geq0$ of the internal microwork $\Delta_{\text{i}}W_{k}$ done by
$\mathfrak{m}_{k}$ that is dissipated; this microwork is different from the
exchange microwork $\Delta_{\text{e}}W$ with the vacuum, which vanishes. We
also establish that $\Delta_{\text{i}}W_{k}\geq0$ for free expansion, which is
remarkable, since its sign is not fixed in a general process. \newline


\end{abstract}
\date{January 31, 2020}
\maketitle

\section{Introduction\label{Sec-Introduction}}

\subsection{Background \label{Sec-WhyThisStudy}}

Free (unrestricted) expansion is an undergraduate paradigm of irreversibility,
in which the exchange macrowork $\Delta_{\text{e}}W$ and macroheat
$\Delta_{\text{e}}Q$, see Fig. \ref{Fig_System}, are identically zero. It is
also accompanied by an increase in the volume $\left\vert \Gamma\right\vert $
of the phase space $\Gamma$ of the system $\Sigma$. One can study it as an
irreversible process $\mathcal{P}$ going on within an isolated system from an
initial (in) macrostate $\mathfrak{M}_{\text{in}}$ to a final (fin) macrostate
$\mathfrak{M}_{\text{fin}}$. This ensures that its energy remains constant
even if the system remains out of equilibrium (EQ) during the entire process
including the initial and final macrostates. The expansion is sudden at $t=0$,
but it takes a while ($t=\tau_{\text{eq}}>0$) for EQ\ to emerge. As is known,
a macrostate $\mathfrak{M}$ of $\Sigma$\ refers to a collection $\left\{
\mathfrak{m}_{k},E_{k},p_{k}\right\}  $ of its microstates $\mathfrak{m}_{k}$
of energies $E_{k}$\ that appear with probabilities $p_{k}$ in $\mathfrak{M}$;
the $p_{k}$'s give rise to stochasticity required for a proper thermodynamics.
We use \emph{macro}- and \emph{micro}- in this study to refer to quantities
pertaining to macrostates and microstates, respectively, with a macroquantity
referring to the \emph{thermodynamic} average of related microquantities. A
microquantity will always carry a subscript $k$ as a reminder that it is
associated with $\mathfrak{m}_{k}$. We will use $Z$ for a state variable (see
Sec. \ref{Sec-InternalVariables-Hamiltonian} for explanation), a
macrovariable, and $Z_{k}$ for its value, a microvariable, for $\mathfrak{m}%
_{k}$. \ 

The study of a particular form of the free expansion is well known at the
undergraduate level in the traditional macroscopic nonequilibrium (NEQ)
thermodynamics \cite{Gibbs,Fermi,Woods,Landau,Prigogine,Kestin} based on the
above exchange quantities; see also \cite{Eu,Jou,Ottinger,deGroot} for modern
treatment. We denote this traditional NEQ thermodynamics by \r{M}NEQT in the
following; here M stands for macroscopic and the small circle refers to the
use of the exchange quantities. The study works if and only if the initial and
final macrostates $\mathfrak{M}_{\text{in,eq}}$ and $\mathfrak{M}%
_{\text{fin,eq}}$, respectively, are in EQ so that the entropy change $\Delta
S=S_{\text{fin}}-S_{\text{in}}$ and, therefore, the net irreversible entropy
change $\Delta_{\text{i}}S=\Delta S$ over the process can be evaluated without
knowing the entire history. But the \r{M}NEQT does not provide any information
during the relaxation ($t<\tau_{\text{eq}}$) towards $\mathfrak{M}%
_{\text{fin,eq}}$ such as the irreversible entropy generation $d_{\text{i}%
}S(t)$ associated with any segment $\delta\mathcal{P}$ of the process between
intermediate NEQ macrostates $\mathfrak{M}(t)$. Thus, the use of the \r{M}NEQT
is limited in its scope.

A NEQ process $\mathcal{P}$ undergoes dissipation at all times $t<\tau
_{\text{eq}}$, and is usually described by the dissipated work $\Delta
_{\text{i}}W>0$, which in turn is directly related to $\Delta_{\text{i}}S$
over $\mathcal{P}$ under suitable conditions; see later. Here, free expansion
poses another hurdle as the common understanding is that any internal work
done by the "vacuum" (absence of matter and radiation) into which the gas
expands must be zero; see Fig. \ref{Fig_Expansion}. This makes it hard to
understand what work is being dissipated as the gas most certainly generates
irreversible entropy $\Delta_{\text{i}}S>0$. A central aspect of this
investigation is to obtain a better understanding of dissipated work
$d_{\text{i}}W$ and the source of $d_{\text{i}}S$ over $\delta\mathcal{P}$ in
an interacting and an isolated system; see Corollary
\ref{Corollary-IsolatedSystem}. This is achieved by focusing on
\emph{system-intrinsic} (SI) quantities $dZ,dZ_{k}$ (which we now allow to
also include $dW,dW_{k}$ and $dQ,dQ_{k}$, which should not be confused with
their exchange analogs $d_{\text{e}}W,d_{\text{e}}W_{k}$ and $d_{\text{e}%
}Q,d_{\text{e}}Q_{k}$; see below) that are \emph{uniquely} determined by the
system itself. They contain all the information including the one about
internal processes that we wish to understand. The exchange quantities
$d_{\text{e}}Z,d_{\text{e}}Z_{k}$ (which also include $d_{\text{e}%
}W,d_{\text{e}}W_{k}$ and $d_{\text{e}}Q,d_{\text{e}}Q_{k}$) are primarily
determined by the macrostate of the medium $\widetilde{\Sigma}$; we will refer
to them as \emph{medium-intrinsic} (MI) quantities here. They are easily
determined by focusing on the medium, which is always taken to be in EQ. Thus,
we can determine $d_{\text{i}}Z\doteq dZ-d_{\text{e}}Z,d_{\text{i}}Z_{k}\doteq
dZ_{k}-d_{\text{e}}Z_{k}$ that directly describe the irreversibility in the
system.%
\begin{figure}
[ptb]
\begin{center}
\includegraphics[
height=3.4714in,
width=3.3044in
]%
{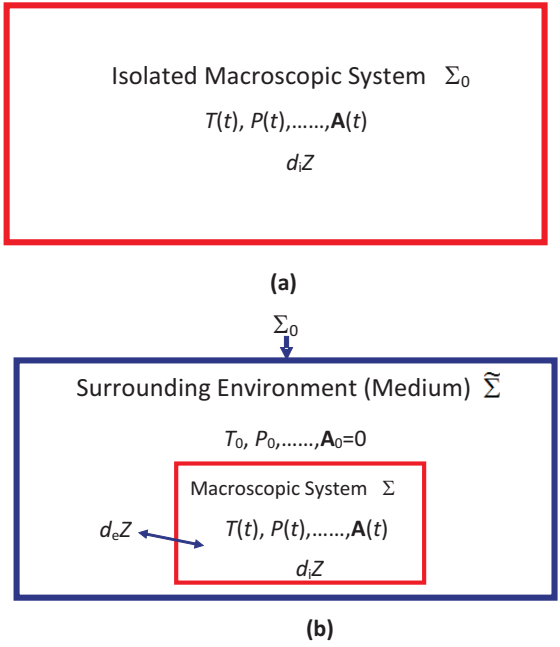}%
\caption{{}(a) An isolated nonequilibrium system $\Sigma_{0}$\ with internally
generated $d_{\text{i}}Z$ driving it towards equilibrium, during which its
SI-fields $T(t),P(t),\cdots,\mathbf{A}(t)$ associated with SI-variables
$S(t),V(t),\cdots,\mathbf{\xi}(t)$ continue to change towards their
equilibrium values; $d_{\text{i}}Z_{k}$ denote the microanalog of
$d_{\text{i}}Z$. The sign of $d_{\text{i}}Z$ is determined by the second law.
(b) A nonequilibrium systen $\Sigma$ in a surrounding medium $\widetilde
{\Sigma}$, both forming the isolated system $\Sigma_{0}$. The macrostates of
the medium and the system are characterized by their fields $T_{0}%
,P_{0},...,\mathbf{A}_{0}=0$ and $T(t),P(t),...,\mathbf{A}(t)$, respectively,
which are different when the two are out of equilibrium. Exchange quantities
($d_{\text{e}}Z$) carry a suffix "e" and irreversibly generated quantities
($d_{\text{i}}Z$) within the system by a suffix "i" by extending the Prigogine
notation. Their sum $d_{\text{e}}Z+d_{\text{i}}Z$ is denoted by $dZ$, which is
a system-intrinsic quantity (see text). In a nonequilibrium system, the
nonzero differences $F_{\text{t}}^{\text{h}}=T-T_{0}$ and $\mathbf{F}%
_{\text{t}}^{\text{w}}=(P-P_{0},\cdots,\mathbf{A})$ denote the set of
thermodynamic forces, where we have also included the affinity $\mathbf{A}%
$\ for internal variables $\boldsymbol{\xi}$; see text. A microstate
$\mathfrak{m}_{k}$\ of $\Sigma$\ is specified by appending a subscript $k$\ to
$\mathbf{F}_{\text{t}}^{\text{w}}$ so that $\mathbf{F}_{\text{t,}k}^{\text{w}%
}=(P_{k}-P_{0},\cdots,\mathbf{A}_{k})$. }%
\label{Fig_System}%
\end{center}
\end{figure}

We have recently developed a version of nonequilibrium thermodynamics (NEQT)
that is expressed in terms of only SI quantities so we have a direct access to
$d_{\text{i}}Z$ and $d_{\text{i}}Z_{k}$. It has appeared in a series of papers
\cite{Gujrati-I,Gujrati-II,Gujrati-III,Gujrati-Entropy1,Gujrati-GeneralizedWork}
covering separate aspects, and reviewed in
\cite{Gujrati-Entropy2,Guj-entropy-2018}. We have labeled it MNEQT to
distinguish it from the \r{M}NEQT. It is briefly introduced in Sec.
\ref{Sec-MNEQT}. The theory is applicable to systems that are either isolated
or in a medium within the same framework. The corresponding microstate version
of the MNEQT is called the $\mu$NEQT, with $\mu$- referring to the use of SI
microquantities. It is capable of studying expansion/contraction at the
microstate level in interacting and isolated system that has not been possible
so far as we will discuss shortly. Another reason to focus on the problem of
expansion/contraction in the $\mu$NEQT\ is due to its close connection with
Maxwell's demon and Landauer's eraser. Both versions of the theory also
involve \emph{internal variables}
\cite{Maugin,Kestin,Prigogine,deGroot,Coleman} that are required, see Sec.
\ref{Sec-InternalVariables}, to explain nonequilibrium internal processes.
Thus, they provide a very general framework of NEQT to understand a majority
of nonequilibrium processes as we will explain.

We know that the classical thermodynamics is based on the concept of work and
heat so we need to identify them in a NEQ process to make any progress. The
central concept in the MNEQT is that of the generalized SI macrowork, see Fig.
\ref{Fig_System}, $dW=P(t)dV(t)+\cdots+\mathbf{A(}t)\cdot d\mathbf{\xi(}t)$
and SI macroheat $dQ=T(t)dS(t)$, see Eq. (\ref{GeneralizedQuantities}), that
are different from the (exchange) MI macrowork $d_{\text{e}}W=P_{0}%
dV(t)+\cdots$ and MI macroheat $d_{\text{e}}Q=T_{0}d_{\text{e}}S$,
respectively, see Eq. (\ref{d-exch}), by irreversible contributions:
\begin{subequations}
\label{GeneralizedIrreversibleQuantities}%
\begin{align}
d_{\text{i}}W  &  =dW-d_{\text{e}}W\geq0,\label{GeneralizedWork}\\
d_{\text{i}}Q  &  =dQ-d_{\text{e}}Q\geq0. \label{GeneralizedHeat}%
\end{align}
The ability to directly deal with $d_{\text{i}}W$ and $d_{\text{i}}Q$ makes
the MNEQT not only perfectly suited to study isolated systems as we will do,
but also ensures that the generalized macrowork $dW$ is isentropic and that
the generalized macroheat $dQ$ satisfies the Clausius identity $dQ=TdS$, see
Eq. (\ref{GeneralizedHeat-General}), in all processes that we are interested
in here; $S$ is always the Gibbs statistical entropy \cite{Gibbs,Landau}%
\end{subequations}
\begin{equation}
S\doteq-%
{\textstyle\sum\nolimits_{k}}
p_{k}\ln p_{k}. \label{StatisticalEntropy}%
\end{equation}
\qquad

The $\mu$NEQT was first introduced a while back \cite{Gujrati-GeneralizedWork}
and applied to a few simple examples including a brief application to the
Brownian motion with a goal to compare its predictions with those from the
work fluctuation theorem (WFT) due to Jarzynski \cite{Jarzynski}; see Eq.
(\ref{JarzynskiRelation}) for its precise formulation. The importance of
\emph{microforce imbalance} $\mathbf{F}_{\text{t,}k}^{\text{w}}$, see Fig.
\ref{Fig_System} caption and later, between externally applied macroforce and
internally generated microforce was pointed out there for the first time. It
is \emph{ubiquitous} in nature
\cite{Gujrati-GeneralizedWork,Gujrati-LangevinEq} as it is always present in
all (EQ and NEQ) macrostates. The macroforce imbalance $\mathbf{F}_{\text{t}%
}=(F_{\text{t}}^{\text{h}},\mathbf{F}_{\text{t}}^{\text{w}})$ between the
fields of the system and the medium, see Fig. \ref{Fig_System} caption,
determines irreversible contribution $\left(  d_{\text{i}}Q,d_{\text{i}%
}W\right)  $ and is well defined even for an isolated system. It vanishes only
in EQ. This makes $\mathbf{F}_{\text{t}}$ and $\mathbf{F}_{\text{t,}%
k}^{\text{w}}$ central in the $\mu$NEQT, which has recently been applied
\cite{Gujrati-LangevinEq} to study the Brownian motion in full detail, where
the relative motion of the Brownian particle with respect to the medium
generates $\mathbf{F}_{\text{t,}k}^{\text{w}}$. Thus, the $\mu$NEQT is also
capable of tackling small systems like Brownian particles under NEQ\ conditions.

\subsection{Why A New Approach? \label{Sec-Prolog}}

Our goal here is to use the $\mu$NEQT and the MNEQT to study general
irreversible processes in interacting and isolated systems with emphasis on
those undergoing phase space volume change and the resulting irreversibility
at a deeper, microscopic level in terms of microstates $\left\{
\mathfrak{m}_{k}\right\}  $. The forthcoming demonstration of the success of
our approach for free expansion, which has not been studied so far, shows its
usefulness as a general theory for both interacting and isolated systems. As a
general setup, we consider an interacting NEQ system $\Sigma$ in a very large
medium $\widetilde{\Sigma}$, see Fig. \ref{Fig_System}(b). The two form an
isolated system $\Sigma_{0}=\Sigma\cup\widetilde{\Sigma}$. Quantities
pertaining to $\Sigma_{0}$ carry a suffix $0$, those pertaining to
$\widetilde{\Sigma}$ carry a tilde, and those pertaining to $\Sigma$ carry no
suffix. For example, the macroworks are $dW_{0},d\widetilde{W}$, and $dW$,
respectively. The medium, being in EQ at all times, has no irreversibility in
it so that $d_{\text{i}}\widetilde{Q}=d_{\text{i}}\widetilde{W}=0$. Because of
its large size, its temperature, pressure, etc. are the same as for
$\Sigma_{0}$ so they are denoted by $T_{0},P_{0}$, etc. as seen in Fig.
\ref{Fig_System}(b).

The \emph{thermodynamic macroforce} \cite{Prigogine,Maugin} $\mathbf{F}%
_{\text{t}}=(T-T_{0},P-P_{0},\cdots\mathbf{)}$ must be nonzero in a NEQ
macrostate\emph{ }and vanish only in an EQ macrostate, \textit{i.e.}, when
$\Sigma$ is in EQ by itself or with $\widetilde{\Sigma}$, if the latter is
present. But the microforce $\mathbf{F}_{\text{t,}k}^{\text{w}}=(P_{k}%
-P_{0},\cdots\mathbf{)}$ is \emph{ubiquitous} as noted above in \emph{all}
macrostates but independent of them, i.e., of $\left\{  p_{k}\right\}  $.
Unfortunately, as we will see, this is not always enforced in many current
microscopic approaches to NEQ thermodynamics.

Care must be exercised if the medium is not extremely large such as in Fig.
\ref{Fig_Expansion}.

Our methodology in the $\mu$NEQT will ensure that the microforces are
\emph{always} accounted for. Given $\mathbf{F}_{\text{t,}k}$, the choice of
$\left\{  p_{k}\right\}  $ determines whether $\mathbf{F}_{\text{t}}=0$ or not
so the methodology will describe thermodynamics correctly. The temporal
development of $\mathfrak{M}$ in any $\mathcal{P}$ can also be studied by
following the deterministic Hamiltonian evolution along \emph{Hamiltonian
trajectories} $\left\{  \gamma_{k}\right\}  $ of microstates $\left\{
\mathfrak{m}_{k}\right\}  $ described by its Hamiltonian $\mathcal{H}$. The
trajectories, therefore, describe deterministic evolution during which
$\left\{  p_{k}\right\}  $ does not change. As $dW$ is isentropic, the
evolution involves the performance of microworks $dW_{k}$ at fixed $p_{k}$;
see later. The stochasticity is due to microheat $dQ_{k}$\ that modifies
$p_{k}$. Thus, $dW_{k}$ and $dQ_{k}$ control different aspects of the
evolution in $\mathfrak{M}$ so their combined effect completes the stochastic
evolution in the $\mu$NEQT.

The trajectories have been recently popularized by modern fluctuation theorems
(MFTs) \cite{Seifert,Broeck}; see also \cite{Blau,Ritort0}. Among these is the
Jarzynski's WFT \cite{Jarzynski}, which is the most celebrated one for the
simple reason that the other MFTs are related to it; see for example, Ref.
\cite{Gujrati-JensenInequality}. Thus, we will comment mostly on the WFT,
commonly known as the JE, in the following, but the comments are equally valid
for other MFTs.

There are four important and independent aspects that require careful
consideration here.

(i) \textbf{Internal variables. }The importance of internal variables
\cite{Prigogine,Kestin,deGroot} and their affinities to describe NEQ
macrostates has been well documented and is an integral part of the MNEQT and
$\mu$NEQT used in this study; see also
\cite{Gujrati-GeneralizedWork,Gujrati-LangevinEq}. We will give a simple
argument for their relevance and the significance of affinities in Sec.
\ref{Sec-InternalVariables}.

(ii) \textbf{Nonequilibrium Entropy}. The MI $d_{\text{e}}Z$ alone provides no
insight into $d_{\text{i}}Z$\ during relaxation unless SI $dZ$ is also
identified. This creates a problem as $\mathfrak{M}(t)$'s denote NEQ
macrostates in general so the SI $dS$ is not known if $S$ is defined only for
EQ macrostates. Thus, we need to identify $S$ for NEQ macrostates. We have
shown that for a NEQ system that is in \emph{internal equilibrium}, the
statistical entropy given in Eq. (\ref{StatisticalEntropy}) is a state
function in an enlarged state space involving internal variables
\cite{Gujrati-I,Gujrati-II}; see Eq. (\ref{StatisticalEntropy-Explicit}). It
is then used in the MNEQT to determine the irreversible contributions
directly. We see from Eq. (\ref{diS}) that $\mathbf{F}_{\text{t}}$ is an
integral part of the MNEQT as promised. We then use the MNEQT to derive the
$\mu$NEQT.

(iii)\textbf{ Phase space volume change }$\Delta\left\vert \Gamma\right\vert
\neq0$. As the number of microstates depends on $\left\vert \Gamma\right\vert
$, there cannot be a one-to-one mapping between the sets of microstates in the
two phase spaces in a process of expansion/contraction. The same problem
arises if $d\left\vert \Gamma\right\vert /dt\neq0$ even if at the end
$\Delta\left\vert \Gamma\right\vert =0$ such as in a cyclic process.

(iv) \textbf{Dissipated work. }We need to provide a physical explanation of
the macrowork that is being dissipated in the free expansion (see Corollary
\ref{Corollary-IsolatedSystem}) and the corresponding microworks.

As interacting systems are also included in our analysis, we make a few
comments in passing about the MFTs, with special attention to the WFT, that
are derived for interacting systems and where trajectories are also exploited.
The formulation invariably uses exchange quantities $\Delta_{\text{e}}W$ and
$\Delta_{\text{e}}Q$\ directly but $\Delta_{\text{i}}W$ and $\Delta_{\text{i}%
}Q$\ are not part of the formulation. Our comments basically summarize the
results already available in the literature.

The MFTs are claimed to describe NEQ processes, because of which they have
attracted a lot of attention. However, despite being part of an highly active
field, we find that they do not provide a useful methodology for our NEQ
consideration here. There is no direct proof of for their NEQ nature that we
are aware. The only indirect proof for the WFT is through the application of
the Jensen's inequality to demonstrate its compliance with the second law
since the inequality leads to
\begin{equation}
\Delta_{\text{e}}W_{\text{J}}\leq-\Delta\overline{F},
\label{Jarzynski-Macrowork}%
\end{equation}
where $\Delta_{\text{e}}W_{\text{J}}$ ("J" for Jarzynski's formulation)\ is a
particular "average exchange" work (properly defined in Eq. (\ref{Del_e-W_J})
later) that is obtained by using the initial probability over the entire
process. It turns out to be a \emph{non-thermodynamic average}
\cite{Gujrati-JensenInequality}, and $\Delta\overline{F}$ is the difference of
the equilibrium (and, therefore, thermodynamic) Helmholtz free energies in the
process; see also comment (f) below in the subsection. The above inequality
looks very similar to the following \emph{thermodynamic} inequality involving
thermodynamic average (exchange macrowork) $R=-\Delta_{\text{e}}W$, where $R$
is the exchange work $\Delta_{\text{e}}\widetilde{W}$ done by $\widetilde
{\Sigma}$ on $\Sigma$,
\begin{equation}
R\geq\Delta\overline{F},\text{ } \label{Exchange-Macrowork}%
\end{equation}
a well-known consequence of the second law, but only if $T_{0}$ remains a
constant in the process \cite{Landau}. In the latter case, the dissipated work
defined as%
\begin{equation}
\Delta_{\text{diss}}W\doteq R-\Delta\overline{F}=T_{0}\Delta_{\text{i}}S\geq0.
\label{DissiaptedWork}%
\end{equation}
To provide an "indirect proof" that the JE is a nonequilibrium result,
Jarzynski sets without any proof that
\begin{equation}
\Delta_{\text{e}}W_{\text{J}}\overset{\text{conjecture}}{=}-R
\label{JarzynskiConjecture}%
\end{equation}
to turn Eq. (\ref{Jarzynski-Macrowork}) into Eq. (\ref{Exchange-Macrowork}).
However, as shown recently \cite{Gujrati-JensenInequality}, Jensen's
inequality applied to the MFTs does not prove compliance with the second law
inequality so $\Delta_{\text{e}}W_{\text{J}}$ in Eq.
(\ref{Jarzynski-Macrowork}) \emph{cannot} be equated with $\Delta_{\text{e}}W$
even when $T_{0}=$ $const$.

There are other concerns about the MFTs, which raise doubts about their
usefulness for our investigation. (a) They do not include any internal
variables, necessary for irreversibility; see Sec. \ref{Sec-BasicConcepts}.
(b) The external macroforce (such as the pressure $P_{0}$) is always assumed
to be equal to the macroforce (such as the pressure $P$) in the system; hence,
they implicitly assume that $d_{\text{e}}W=dW$, which results in $d_{\text{i}%
}W\equiv0$; see Eq. (\ref{diW}). This was first pointed out in Ref.
\cite{Gujrati-GeneralizedWork}. Thus, they do not include any thermodynamic
macroforce $\mathbf{F}_{\text{t}}^{\text{w}}$ necessary for $d_{\text{i}}W$
and for irreversibility \cite{Prigogine}. (c) From $d_{\text{i}}W\equiv0$
follows $d_{\text{i}}Q\equiv0$,\ see Eq. (\ref{diQ-diW}). If the temperature
of the system is always equal to $T_{0}$, \textit{i.e.}, $F_{\text{t}%
}^{\text{h}}=T-T_{0}=0$, then it follows from Eqs. (\ref{diQ}, \ref{diS}) that
$d_{\text{i}}S\equiv0$. Cohen and Mauzerall \cite{Cohen0,Cohen00} were the
first to raise concern that $T$ may not even exist for a NEQ process; see also
\cite{Jarzynski-Cohen} for counter-arguments, some of which we will discuss
later. The concern was justified as correct later by Muschik \cite{Muschik} so
$T=T_{0}$ will make the MFTs unsuitable for a NEQ process. (d) MFTs are based
on a \emph{fixed} set of classical microstates $\left\{  \mathfrak{m}%
_{k}\right\}  $ or trajectories $\left\{  \gamma_{k}\right\}  $ as the use of
Hamiltonian dynamics is consistently prevalent. Thus, their applicability is
limited to the situation $d\left\vert \Gamma\right\vert =\Delta\left\vert
\Gamma\right\vert =0$; see Sec. \ref{Sec-VoluemChange}. This was first pointed
out by Sung \cite{Sung}. Unfortunately, this limitation is not well recognized
in the field. (e) The WFT should also apply to an isolated system
\cite{JarzynskiNote}. Because $d_{\text{e}}W=0$ in this case, they do not. (f)
The free expansion in Fig. \ref{Fig_Expansion} refers to an isolated system so
the WFT should be applicable in this case \cite{KestinNote}, but does not as
first pointed out by Sung \cite{Sung}; see also \cite{Gross,Jarzynski-Gross}
for the ensuing debate. We will come back to this issue later when we discuss
free expansion. (g) In addition, the averaging in the WFT is not a
thermodynamic averaging over the process as first hinted by Cohen and
Mauzerall \cite{Cohen0}, and established rigorously recently by us
\cite{Gujrati-JensenInequality}, whereas we require a thermodynamic averaging
in our investigation.

Because of all these limitations, the MFTs are not of central interest to us
in this study except to draw attention to the differences with our approach.
Therefore, we will discuss and substantiate the above points again later in
Sec. \ref{Sec-WFT} within our theoretical framework; we focus on the WFT for simplicity.

There have been several numerical attempts to study restricted expansion in an
interacting system \cite{Lua0,Lua1,Baule,Crooks-Jarzynski,Davie,Jarzynski1}
but with a goal only to verify the WFT. Because of this, these numerical
studies are also not helpful to us for the reasons stated above.

In conclusion, it is not a surprise that we are left to exclusively use the
MNEQT and $\mu$NEQT in this study of interacting and isolated systems. We have
already applied the MNEQT to briefly study free expansion
\cite{Gujrati-Heat-Work0,Gujrati-Heat-Work}. Here, we wish to go beyond the
earlier study to demonstrate how the $\mu$NEQT can be used to study
expansion/contraction with special attention to free expansion by including
internal variables also. The $\mu$NEQT has also been recently applied
successfully to provide a \emph{thermodynamic alternative} to study Brownian
motion without using the mechanical approach involving the Langevin equation
\cite{Gujrati-LangevinEq}.\ The macroscopic friction emerges as a consequence
of the relative motion of the Brownian particle, an internal variable, with
respect to the medium. We do not need to postulate the Langevin noise term; it
emerges as a consequence of thermodynamic averaging mentioned above.

Our methodology and theory will be formulated for any arbitrary process in
both interacting and isolated systems. The theory is derived from the MNEQT so
it is always consistent with classical thermodynamics. The process will also
include expansion and contraction as special cases but the main focus will be
mostly on the spontaneous process of unrestricted, \textit{i.e., }free
expansion for the reason explained above. Whenever we study free expansion, we
will consider the gas as a closed system $\Sigma$, which is in a medium
$\widetilde{\Sigma}$ that happens to be the vacuum; see Fig.
\ref{Fig_Expansion}. Their combination forms the isolated system shown by
$\Sigma_{0}$ in Fig. \ref{Fig_System}. For the set up for free expansion, we
follow Kestin \cite{KestinNote} as we want to make the system ($\Sigma$) and
the vacuum ($\widetilde{\Sigma}$) independent. As $\widetilde{\Sigma}$ is
devoid of matter and radiation, $\Sigma_{0}$ is nothing but $\Sigma$. We can
replace the partition by a piston exerting an external pressure $P_{0}$ from
$\widetilde{\Sigma}$ for a general expansion/contraction process. For
$P_{0}<P$, the gas will expand; for $P_{0}>P$, the gas will contract. As the
piston is an insulator, $\widetilde{\Sigma}$ only acts as a working medium
$\widetilde{\Sigma}_{\text{w}}$. We need to bring in a thermal medium
$\widetilde{\Sigma}_{\text{h}}$ to bring about thermal equilibrium. Such an
expansion/contraction process is covered by the WFT \cite{Jarzynski}. As
$P_{0}\rightarrow0$, and $P>0$, we obtain the limiting case of free expansion.
Thus, free expansion is merely a limiting case of expansion/contraction in our
approach, and does not require a separate approach. In all these cases, we
require the gas particles initially to be always confined in the left
chamber.
\begin{figure}
[ptb]
\begin{center}
\includegraphics[
height=2.386in,
width=3.2353in
]%
{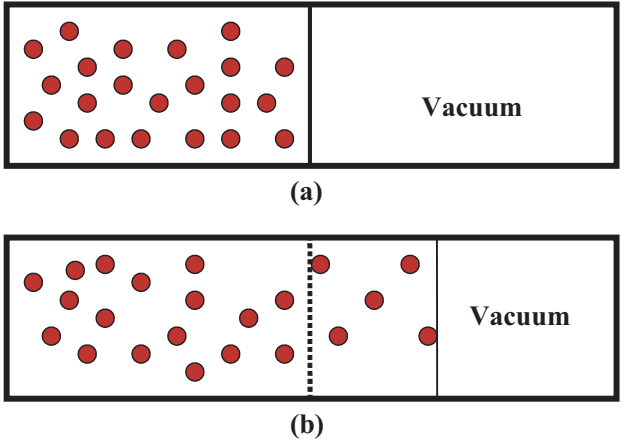}%
\caption{Free expansion of a gas. The gas is confined to the left chamber,
which is separated by a hard partition (shown by a solid black vertical line)
from the vacuum in the right chamber as shown in (a). At time $t=0$, the
partition is removed abruptly as shown by the broken line in its original
place in (b). The gas expands in the empty space, devoid of matter and
radiation, on the right but the expansion is gradual as shown by the solid
front, which separates it from the vacuum on its right. We can also think of
the hard partition in (a) as a piston, which maintains the volume of the gas
on its left. The piston can be moved slowly or rapidly to the right within the
right chamber with a pressure $P_{0}<P$ to change this volume. The free
expansion occurs when the piston moves extremely (infinitely) fast by letting
$P_{0}\rightarrow0$.}%
\label{Fig_Expansion}%
\end{center}
\end{figure}

\subsection{Layout}

We briefly review some useful new concepts in the next section. We begin the
discussion with the need for internal variables in a NEQ\ process. The nature
of the parameters in the Hamiltonian of a NEQ system is discussed after that,
which is then followed by the definition of Hamiltonian trajectories. We close
this section by introducing the central concept of internal equilibrium. We
follow this section by a brief introduction to the MNEQT in Sec.
\ref{Sec-MNEQT}. Here, we show the importance of internal variables for an
isolated system to determine the irreversible macrowork; see Theorem
\ref{Theorem-diW-diS-Isolated}. This partially answers one of the motivating
questions. This is then followed by an introduction to the $\mu$NEQT in Sec.
\ref{Sec-muNEQT}, where we introduce the concepts of various microworks and
microheats. We introduce the moment generating function in Sec. \ref{Sec-MGF0}%
, which gives all the moments including fluctuations of microworks from this
single function. We then turn to the free expansion of a classical gas and
study it by restricting to use only two internal variables in the MNEQT in
Sec. \ref{Sec-FreeExpansion-MNEQT}. We study the free expansion in the $\mu
$NEQT in the next section. We first study it in a quantum case in Sec.
\ref{Sec-QuantumExpansion} and then in the classical case in Sec.
\ref{Sec-ClassicalExpansion}, where we introduce the important trick that
allows us to consider free expansion for any arbitrary expansion. The moment
generating function is used to directly demonstrate that the trick does not
affect thermodynamics. The trick can be extended to a cyclic process during
which the phase space volume changes nonmonotonically or to restricted
expansion and contraction. A brief discussion of our results is presented in
the last section. \ 

\section{Basic Concepts\ \label{Sec-BasicConcepts}}

\subsection{Need for an Internal Variable \label{Sec-InternalVariables}}

Consider a simple example of a NEQ system of $N$ particles, each of which can
be in two levels, forming an isolated system $\Sigma$ of volume $V$. Let
$\rho_{l}$ and $e_{l}(V),l=1,2$ denote the probabilities and energies of the
two levels of a particle in a NEQ macrostate so that $\rho_{1},\rho_{2}$ keep
changing. We have $e=\rho_{1}e_{1}+\rho_{2}e_{2}$ for the average energy per
particle, which is a constant, and $d\rho_{1}+d\rho_{2}=0$ as a consequence of
$\rho_{1}+\rho_{2}=1$. Using $de=0$, we get
\[
d\rho_{1}+d\rho_{2}e_{2}/e_{1}=0,
\]
which, for $e_{1}\neq e_{2}$, is inconsistent with the second equation (unless
$d\rho_{1}=0=d\rho_{2}$, which corresponds to EQ). Thus, $e_{l}(V)$ cannot be
treated as constant in determining $de$. In other words, there must be an
extra dependence in $e_{l}$ so that%
\[
e_{1}d\rho_{1}+d\rho_{2}e_{2}+\rho_{1}de_{1}+\rho_{2}de_{2}=0,
\]
and the inconsistency is removed. This extra dependence must be due to
\emph{independent} internal variables that are not controlled from the outside
(isolated system) so they continue to relax in $\Sigma$ as it approaches EQ.
Let us imagine that there is a single internal variable $\xi$ so that we can
express $e_{l}$ as $e_{l}(V,\xi)$ in which $\xi$ continues to change as the
system comes to equilibrium. The above equation then relates $d\rho_{1}$ and
$d\xi$; they both vanish simultaneously as EQ is reached. We also see that
without any $\xi$, the isolated system cannot equilibrate.

The above discussion is easily extended to a $\Sigma$ with many energy levels
of a particle with the same conclusion that at least a single internal
variable is required to express $e_{l}=e_{l}(V,\xi)$ for each level $l$. We
can also visualize the above system in terms of microstates. A microstate
$\mathfrak{m}_{k}$ refers to a particular distribution of the $N$ particles in
any of different levels with energy $E_{k}=%
{\textstyle\sum\nolimits_{l}}
N_{l}e_{l}$, where $N_{l}$ is the number of particles in the $l$th level, and
is obviously a function of $N,V,\xi$ so we will express it as $E_{k}(N,V,\xi
)$. This makes the average energy of the system also a function of $N,V,\xi$,
which we express as $E(N,V,\xi)$.

An EQ system is uniform. Thus, the presence of $\xi$ suggests some sort of
nonuniformity in the system. To appreciate its physics, we consider a slightly
different situation below as a possible example of nonuniformity.

We consider as a simple NEQ example a composite isolated system $\Sigma$
consisting of two identical\ subsystems $\Sigma_{1}$ and $\Sigma_{2}$ of
identical volumes and numbers of particles but at different temperatures
$T_{1}$ and $T_{2}$ at any time $t\leq\tau_{\text{eq}}$ before EQ is reached
at $t=\tau_{\text{eq}}$ so the subsystems have different time-dependent
energies $E_{1}$ and $E_{2}$, respectively. We assume a diathermal wall
separating $\Sigma_{1}$ and $\Sigma_{2}$. Treating each subsystem in EQ at
each $t$, we write their entropies as $S_{1\text{eq}}(E_{1},V/2,N/2)$ and
$S_{2\text{eq}}(E_{2},V/2,N/2)$. The entropy $S$ of $\Sigma$ is a function of
$E_{1},E_{2},V$, and $N$. Obviously, $\Sigma$ is in a NEQ macrostate at each
$t<\tau_{\text{eq}}$. From $E_{1}$ and $E_{2}$, we form two independent
combinations $E=E_{1}+E_{2}=$constant and $\xi=E_{1}-E_{2}$ so that we can
express the entropy as $S(E,V,N,\xi)$. Here, $\xi$ plays the role of an
internal variable, which continues to relax\ towards zero as $\Sigma$
approaches EQ. For given $E$ and $\xi$,\ $S(E,V,N,\xi)$ has the \ maximum
possible values since both $S_{1\text{eq}}$ and $S_{2\text{eq}}$ have their
maximum value. As we will see below, this is the idea behind the concept of
internal equilibrium in which $S(E,V,N,\xi)$ is a state function of state
variables and continues to increase as $\xi$ decreases and vanishes in EQ.

We assume $\Sigma$ to be in IEQ at each $t$ in this simple example. From
$1/T=\partial S/\partial E$ and $A/T=\partial S/\partial\xi$, $A$ being the
\emph{activity} associated with $\xi$, we find that
\[
T=2T_{1}T_{2}/(T_{1}+T_{2}),A=(T_{2}-T_{1})/(T_{1}+T_{2}).
\]
As EQ is attained, $T\rightarrow T_{\text{eq}}$, the EQ temperature of both
subsystems and $A\rightarrow A_{0}=0$ as expected. We see that in this simple
example $Ad\xi/T$ is the contribution due to irreversiblity in $dS$, which
also shows that $(-Ad\xi)$ is the contribution due to irreversiblity in $dE$.

In general, the activity $\mathbf{A}$ controls the behavior of
$\boldsymbol{\xi}$ in a NEQ macrostate and vanishes when EQ is reached.\ Here,
we will take a more general view of $\mathbf{A}$, and extend its definition to
$\mathbf{X}$ also \cite{Gujrati-I}. Comparing $A$ with $F_{\text{t}}%
^{\text{h}}$, we clearly see that $F_{\text{t}}^{\text{h}}$ also plays the
role of an activity. The same reasoning also shows that $\mathbf{F}_{\text{t}%
}^{\text{w}}$ plays the role of an activity.

The example can be easily extended to the case of expansion and contraction by
replacing $E,E_{1}$, and $E_{2}$ by $N,N_{\text{L}}$, and $N_{\text{R}}$, see
Fig. \ref{Fig_Expansion}, to describe the diffusion of particles. The role of
$\beta$ and $E$, etc. are played by $\beta\mu$ and $N$, etc.

\subsection{Hamiltonian with Internal
Variables\label{Sec-InternalVariables-Hamiltonian}}

It is clear that in order to capture a NEQ process, internal variables are
\emph{necessary}. Another way to appreciate this fact is to realize that for
an isolated system, all the observables in $\mathbf{X=(}E,V,N,\cdots
\mathbf{)}$ are fixed so if the entropy is a function of $\mathbf{X}$ only, it
cannot change
\cite{Gujrati-I,Gujrati-II,Guj-entropy-2018,Gujrati-Entropy1,Gujrati-Entropy2}%
. Thus, we need additional independent variables to ensure the law of increase
of entropy for a NEQ system. An EQ macrostate is represented by a point in the
state space $\mathfrak{S}_{0}$ spanned by $\mathbf{X}$, but a NEQ macrostate
by a point in an \emph{enlarged} state space $\mathfrak{S}$ spanned by
$\mathbf{Z\doteq}(\mathbf{X,}\boldsymbol{\xi})$, where $\boldsymbol{\xi}$ is
the set formed by internal variables. Internal variables cannot be controlled
from the outside of the system; they are only controlled by the processes
within the system. On the other hand, the observables in $\mathbf{X}$\ are
controlled from the outside. We will call $\mathbf{X}$ a set of observables
and $\mathbf{Z}$ a set of state variables. In EQ, internal variables are no
longer independent of the observables. Consequently, their affinities (see
later) vanish in EQ. It is common to define the internal variables so their EQ
values vanish.

As we will be dealing with the Hamiltonian of the system, it is useful to
introduce the sets $\mathbf{X}^{E}\doteq\mathbf{X\backslash}E=(V,N,\cdots)$,
and $\mathbf{Z}^{E}\doteq\mathbf{Z\backslash}E=(V,N,\cdots,\boldsymbol{\xi
})=(\mathbf{X}^{E}\mathbf{,}\boldsymbol{\xi})$. Then, $E$ and $E_{k}$\ become
a function of $\mathbf{Z}^{E}$ as we saw in Sec. \ref{Sec-InternalVariables}.
Here, $\mathbf{Z}^{E}$ appears as a \emph{parameter} in the Hamiltonian, which
we will write as $\mathcal{H}(\left.  \mathbf{z}\right\vert \mathbf{Z}^{E})$,
where $\mathbf{z}$ is a point (collection of coordinates and momenta of the
particles) in the phase space $\Gamma(\mathbf{Z}^{E})$ specified by
$\mathbf{Z}^{E}$. As an example, $N,V$, and $\xi$ are the parameters in the
previous section. When the system moves about in the phase space
$\Gamma(\mathbf{Z}^{E})$, $\mathbf{z}$ changes but $\mathbf{Z}^{E}$ as a
parameter remains fixed in a state subspace $\mathfrak{S}^{E}\doteq
\mathfrak{S}\mathbf{\backslash}E$.

\subsection{Hamiltonian Trajectories}

Traditional formulation of statistical thermodynamics
\cite{Landau,Gibbs,Gujrati-Symmetry} takes a mechanical approach in which
$\mathfrak{m}_{k}$ follows its classical or quantum mechanical evolution
dictated by its SI Hamiltonian $\mathcal{H}(\left.  \mathbf{z}\right\vert
\mathbf{Z}^{E})$. The quantum microstates are specified by a set of good
quantum numbers, which we have denoted by $k$ above as a single quantum number
for simplicity; we take $k\in\mathbb{N},\mathbb{N}$ denoting the set of
natural numbers. We will see below that $k$\ does not change as $\mathbf{Z}%
^{E}$ changes. In the classical case, we can use a small cell $\delta
\mathbf{z}_{k}$ of size $(2\pi\hbar)^{3N}$\ around $\mathbf{z}_{k}=\mathbf{z}$
as the microstate $\mathfrak{m}_{k}$. In the rest of the work, we will keep
$N$ fixed to fix the size of the system. Therefore, from now on, $\mathbf{X}$
and $\mathbf{Z}$\ will not contain it. The Hamiltonian gives rise to a purely
mechanical evolution of individual $\mathfrak{m}_{k}$'s, which we will call
the \emph{Hamiltonian evolution}, and suffices to provide their mechanical
description. The change in $\mathcal{H}(\left.  \mathbf{z}\right\vert
\mathbf{Z}^{E})$ in a process is%
\begin{subequations}
\begin{equation}
d\mathcal{H}=\frac{\partial\mathcal{H}}{\partial\mathbf{z}}\cdot
d\mathbf{z}+\frac{\partial\mathcal{H}}{\partial\mathbf{Z}^{E}}\cdot
d\mathbf{Z}^{E}. \label{HamiltonianChange}%
\end{equation}
The first term on the right vanishes identically due to Hamilton's equations
of motion for any $\mathfrak{m}_{k}$. Thus, for fixed $\mathbf{Z}^{E}$, the
energy $E_{k}=\mathcal{H}_{k}\doteq\mathcal{H}(\left.  \mathbf{z}%
_{k}\right\vert \mathbf{Z}^{E})$ remains constant as $\mathfrak{m}_{k}$ moves
about in $\Gamma(\mathbf{Z}^{E})$. Only the variation $d\mathbf{Z}^{E}$ in
$\mathfrak{S}$ generates any change in $E_{k}$. Consequently, we do not worry
about how $\mathbf{z}_{k}$\ changes in $\mathcal{H}(\left.  \mathbf{z}%
\right\vert \mathbf{Z}^{E})$ in the phase space, and focus, instead, on the
state space $\mathfrak{S}$, in which\ can write%
\begin{equation}
dE_{k}=\frac{\partial E_{k}}{\partial\mathbf{Z}^{E}}\cdot d\mathbf{Z}%
^{E}=-dW_{k}, \label{HamiltonianChange-StateSpace}%
\end{equation}
where $dW_{k}$ denotes the\emph{ generalized microwork} produced by the
\emph{generalized microforce} $\mathbf{F}_{k}^{E}$:
\begin{equation}
dW_{k}=\mathbf{F}_{k}^{E}\cdot d\mathbf{Z}^{E},\ \mathbf{F}_{k}^{E}%
\doteq-\partial E_{k}/\partial\mathbf{Z}^{E}. \label{Generalized force-work}%
\end{equation}
We can now identify $\mathbf{Z}^{E}$ as the \emph{work parameter}, whose
variation $d\mathbf{Z}(t)\doteq(dE(t),d\mathbf{Z}^{E}(t))$ in $\mathfrak{S}$
defines not only the microworks $\left\{  dW_{k}\right\}  $, but also\ a
thermodynamic process $\mathcal{P}$. The trajectory $\gamma_{k}$ in
$\mathfrak{S}$ followed by $\mathfrak{m}_{k}$ as a function of time will be
called the \emph{Hamiltonian trajectory} during which $\mathbf{Z}^{E}$ varies
from its initial (in) value $\mathbf{Z}_{\text{in}}^{E}$\ to its final (fin)
value $\mathbf{Z}_{\text{fin}}^{E}$ during $\mathcal{P}$. The variation
produces the generalized microwork $dW_{k}$; $p_{k}$ plays no role so $dW_{k}$
is purely mechanical, which simplifies its determination in our theory. The
microwork $dW_{k}$ also does not change the index $k$ of $\mathfrak{m}_{k}$ as
said above.

Being purely mechanical in nature, a trajectory is completely
\emph{deterministic} and cannot describe the evolution of a macrostate
$\mathfrak{M}$ during $\mathcal{P}$ unless supplemented by thermodynamic
stochasticity, which requires $p_{k}(\mathfrak{M})$ as discussed above
\cite{Landau}, and is related to $dQ_{k}$ as shown later; see Eq.
(\ref{microheat}). Thermodynamics emerges when quantities pertaining to the
trajectories are averaged over the trajectory ensemble $\left\{  \gamma
_{k}\right\}  $ with appropriate probabilities that will usually change during
the process. In this sense, our approach is different from approaches using
stochastic trajectories \cite{Seifert,Broeck}, where $dE_{k}$ is identified
with the exchange microwork $dR_{k}\doteq-d_{\text{e}}W_{k}$; see Remark
\ref{Remark-IncorrectIdentification-Energy}.

The development of the $\mu$NEQT requires pursuing individual trajectory
$\gamma_{k}$ of $\mathfrak{m}_{k}(\mathbf{Z}^{E})$.\ In the following,
$\mathfrak{m}_{k}$ will usually stand for classical microstates unless
specified otherwise, and follows its deterministic trajectory $\gamma_{k}$ as
$\mathfrak{m}_{k\text{,in}}\doteq\mathfrak{m}_{k}(\mathbf{Z}_{\text{in}}^{E})$
evolves into $\mathfrak{m}_{k\text{,fin}}\doteq\mathfrak{m}_{k}(\mathbf{Z}%
_{\text{fin}}^{E})$. By reversing the change in $\mathbf{Z}^{E}$ from
$\mathbf{Z}_{\text{fin}}^{E}$ to $\mathbf{Z}_{\text{in}}^{E}$, $\mathfrak{m}%
_{k}(\mathbf{Z}_{\text{fin}}^{E})$ comes back to $\mathfrak{m}_{k}%
(\mathbf{Z}_{\text{in}}^{E})$. Thus, $\gamma_{k}$ defines an $\emph{identity}$
map $\mathcal{T}:$
\end{subequations}
\begin{equation}
\mathcal{T}:\mathfrak{m}_{k}(\mathbf{Z}_{\text{in}}^{E})\rightleftarrows
\mathfrak{m}_{k}(\mathbf{Z}_{\text{fin}}^{E}) \label{Temporal mapping}%
\end{equation}
without altering the index $k$ so it is a one-to-one ($1$-to-$1$) or identity
mapping of microstates. The two arrows mean that the mapping can be
\emph{inverted} without altering the index $k$ so it is a one-to-one
($1$-to-$1$) or identity mapping of microstates.

\subsection{Phase Space Volume Change \label{Sec-VoluemChange}}

However, Hamiltonian trajectories for classical microstates are not suitable
for processes that involve expansion and contraction in the volume $V$ and/or
other parameters in $\mathbf{Z}^{E}$ of the system with a corresponding change
in the phase space volume $\left\vert \Gamma\right\vert $. In the following,
we will think of $V$ as the varying work parameter for simplicity. Then,
during expansion, the initial volume $\left\vert \Gamma\right\vert
_{\text{in}}\doteq\left\vert \Gamma(\mathbf{Z}_{\text{in}}^{E})\right\vert $
is smaller than the final volume $\left\vert \Gamma^{\prime}\right\vert
_{\text{fin}}\doteq\left\vert \Gamma(\mathbf{Z}_{\text{fin}}^{E})\right\vert $
as shown in Fig. \ref{Fig_PhasePoint-Evolution}. This means that there are
microstates such as $\delta\mathbf{\zeta}^{\prime}$ ($\neq\delta
\mathbf{z}^{\prime}$) in $\Gamma_{\text{fin}}^{\prime}\doteq\Gamma
(\mathbf{Z}_{\text{fin}}^{E})$\ that cannot be reached from any of the
microstate $\delta\mathbf{z}$ in $\Gamma_{\text{in}}\doteq\Gamma
(\mathbf{Z}_{\text{in}}^{E})$ along Hamiltonian trajectories; the latter take
$\delta\mathbf{z}$ into $\delta\mathbf{z}^{\prime}$\ inside the broken
horizontal ellipse. The converse is true for contraction.

Note also that we are not interested in the cardinality of the initial\ and
final sets of microstates $\left\{  \mathfrak{m}_{\text{in}}\right\}  $ and
$\left\{  \mathfrak{m}_{\text{fin}}\right\}  $,respectively. We are interested
in how they map under Hamiltonian evolution; see Sec. \ref{Sec-Conclusion} for
more clarification.

It is clear that in the classical case, we require a new approach to overcome
the loss of the $1$-to-$1$ mapping if we confine ourselves to only Hamiltonian
trajectories. To the best of our knowledge, the problem of how to overcome
this hurdle of phase space volume change using Hamiltonian trajectories has
not been solved.%

\begin{figure}
[ptb]
\begin{center}
\includegraphics[
height=1.9813in,
width=3.3615in
]%
{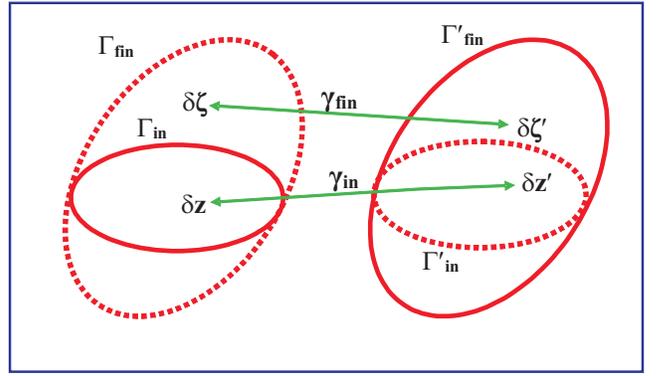}%
\caption{The evolution of a microstate $\delta\mathbf{z}\in\mathbf{\Gamma
}_{\text{in}}$ into $\delta\mathbf{z}^{\prime}\in\mathbf{\Gamma}_{\text{fin}%
}^{\prime}$\ following the variation in $\mathbf{Z}^{E}$ (green arrows). The
initial and final phase spaces are $\mathbf{\Gamma}_{\text{in}}$ and
$\mathbf{\Gamma}_{\text{fin}}^{\prime}$ shown by the interiors of the red
solid ellipses. By changing $\mathbf{Z}^{E}$ in the reverse order maps
$\delta\mathbf{z}^{\prime}$ into $\delta\mathbf{z}$ as implied by the reverse
green arrows. The microstates $\delta\mathbf{\zeta}^{\prime}$ and
$\delta\mathbf{\zeta}$, and other quantities are explained in the text.}%
\label{Fig_PhasePoint-Evolution}%
\end{center}
\end{figure}

We overcome the hurdle by introducing a novel but simple trick. The clue for
the new approach comes from considering trajectories in quantum mechanics. As
quantum microstates form a denumerable set with index $k\in\mathbb{N}$ as
$\mathbf{Z}^{E}$\ changes, there is a $1$-to-$1$ mapping \ as in Eq.
(\ref{Temporal mapping}) between $\mathfrak{m}_{k\text{,in}}$ and
$\mathfrak{m}_{k\text{,fin}}$\ during expansion and contraction, which helps
us remedy the lack of one-to-one correspondence due to volume change in the
classical case. The trick is to enlarge the smaller phase space to become
equal to the larger phase space by adding missing microstates that appear with
nonzero but vanishingly small probabilities. As the work along deterministic
trajectories in Eq. (\ref{HamiltonianChange-StateSpace}) is oblivious to their
probabilities (even though they continue to change in thermodynamics), we can
add trajectories initiating at the missing microstates to obtain a
\emph{enlarged }$1$-to-$1$\emph{ trajectory ensemble} $\left\{  \gamma
\right\}  $. At the end of the computation of ensemble averages, a formal
limit of vanishing probabilities of missing initial microstates is taken.

\subsection{Internal Equilibrium (IEQ)\label{Sec-IEQ}}

The central concept of the NEQT exploited here is that of the \emph{internal
equilibrium }(IEQ) according to which the entropy $S$ of a NEQ macrostate is a
\emph{state function} of the state variables in the enlarged state space
$\mathfrak{S}$ \cite{Gujrati-I,Gujrati-II,Gujrati-III} relative to the EQ
state space $\mathfrak{S}_{0}$ due to independent internal variables
\cite{Coleman,deGroot,Prigogine,Maugin,Gujrati-I,Gujrati-II} that are required
to describe a NEQ macrostate as explained above. In EQ, the internal variables
are no longer independent of the observables forming the space $\mathfrak{S}%
_{0}$. As a consequence, their affinities vanish in EQ. In general, the
temperature $T$ of the system in IEQ is identified in the standard manner by
the relation
\begin{equation}
1/T=\partial S/\partial E \label{IEQ-Temp}%
\end{equation}
using the fact that $S$ is state function in $\mathfrak{S}$.

An important property of IEQ macrostates is the following: It is possible in
an IEQ macrostate to have different degrees of freedom or different parts of a
system to have different temperatures than $T$. For example, in a glass, it is
well known that the vibrational degrees of freedom have a different
temperature than the configurational degrees of freedom
\cite{Debenedetti,Gujrati-Hierarchy}. In the viscous drag problem, the
CM-motion of the Brownian particle can have a different temperature than $T$
of the rest of the particles in the fluid \cite{Gujrati-LangevinEq}. This
observation is easily verified in MNEQT based on the concept of IEQ as done
elsewhere \cite[see Sec. 8.1 and Eq. (58)]{Gujrati-Hierarchy}. By taking a
larger and larger set of internal variables, we can treat almost all NEQ
macrostates as if in IEQ. Thus, the MNEQT is an extremely useful
thermodynamics for NEQ systems.

\section{The MNEQT\label{Sec-MNEQT}}

An EQ macrostate is described by $\mathbf{X=(}E,V,\cdots\mathbf{)}$, and its
entropy is a state function $S(\mathbf{X})$. Away from EQ, $S(\mathbf{Z})$
becomes a state function for NEQ macrostates in IEQ. In the following, we will
focus on $V$ and $\xi$ as members of $\mathbf{Z}^{E}\doteq\mathbf{Z\backslash
}E$ for simplicity but the discussion is general and applies to any
$\mathbf{Z}^{E}$. Indeed, we use two internal variables when we study free
expansion. The microstate $\mathfrak{m}$ follows its evolution dictated by its
(SI) Hamiltonian $\mathcal{H}(\left.  \mathbf{z}\right\vert V,\boldsymbol{\xi
})$; the interaction with $\widetilde{\Sigma}$ is usually treated as a very
weak stochastic perturbation, which immediately suggests adopting a SI description.

\subsection{Macrowork and Macroheat}

There are two kinds of macrowork,the SI and MI macroworks $dW$ and
$d_{\text{e}}W$, respectively, in NEQT; see Eq. (\ref{GeneralizedWork}). The
irreversible macrowork $d_{\text{i}}W\geq0$ vanishes only for a reversible
process. We similarly have two kinds of macroheat, the SI macroheat $dQ$ and
the MI exchange macroheat $d_{\text{e}}Q$; see Eq. (\ref{GeneralizedHeat}).
The irreversible heat $d_{\text{i}}Q\geq0$ vanishes only for a reversible
process. The first law can be equivalently expressed in terms of SI and MI
quantities, respectively:
\begin{subequations}
\label{FirstLaw}%
\begin{align}
dE  &  =dQ-dW,\label{FirstLaw-SI}\\
dE  &  =d_{\text{e}}Q-d_{\text{e}}W; \label{FirstLaw-MI}%
\end{align}
the first equation follows from Eq. (\ref{StateChange-Averages}) and the
discussion following it. From the above two equations follows the important
identity
\end{subequations}
\begin{equation}
d_{\text{i}}Q\equiv d_{\text{i}}W, \label{diQ-diW}%
\end{equation}
which establishes that internal processes ensure that irreversible macroheat
and macrowork within $\Sigma$\ are identically equal in magnitude. This
equality is very general and will be used extensively in this study. It is a
consequence of a very general result of NEQT that no irreversible process can
generate any internal change in the (average) energy of the system,
\textit{i.e.,}
\begin{equation}
d_{\text{i}}E\equiv0, \label{diE}%
\end{equation}
so that $d_{\text{e}}E=dE$. From now on, we will refer to $dW$ and $dQ$
($d_{\text{e}}W$ and $d_{\text{e}}Q$) simply as macrowork and macroheat
(exchange macrowork and exchange macroheat), respectively, so no confusion can arise.

The discussion above is valid for any arbitrary process, but, from now on, we
restrict to the case when $S=S(\mathbf{Z})$ is a state function in
$\mathfrak{S}$ for simplicity, \textit{i.e.}, for the macrostate to be an IEQ
macrostate \cite{Gujrati-I,Gujrati-II} to define $T$. The discussion to an
arbitrary process, which can be done, will be avoided here.

For a process requiring pressure-volume macrowork only, we have $dW=PdV$ done
\emph{by} $\Sigma$ (SI)\ or\emph{ }$d\widetilde{W}=-P_{0}dV$ done \emph{by}
$\widetilde{\Sigma}$ (MI) in terms of the instantaneous pressure $P=-\partial
E/\partial V$\ of $\Sigma$ or $P_{0}~$of $\widetilde{\Sigma}$, and their
volume change $dV$ or $-dV$, respectively. The exchange macrowork is
$d_{\text{e}}W=-d\widetilde{W}=P_{0}dV$. For irreversibility, $P\neq P_{0}$,
with $P-P_{0}$ playing the role of an activity \cite{Gujrati-I}. The
irreversible or dissipated work is $d_{\text{i}}W=(P-P_{0})dV$
\cite{Kestin,Woods,Prigogine}.

Comparing Eq. (\ref{FirstLaw-SI}) for an IEQ macrostate with
\begin{equation}
dS=(\partial S/\partial\mathbf{X})\cdot d\mathbf{X+}(\partial S/\partial
\boldsymbol{\xi})\cdot d\boldsymbol{\xi}, \label{EntropyDiff-Fundamental}%
\end{equation}
allows us to identify the macroheat and macrowork as
\begin{subequations}
\label{GeneralizedQuantities}%
\begin{align}
dW  &  =T(\partial S/\partial\mathbf{Z}^{E})\cdot d\mathbf{Z}^{E}%
=PdV+\cdots+\mathbf{A}\cdot d\boldsymbol{\xi}, \label{GeneralizedWork-General}%
\\
dQ  &  =TdS\boldsymbol{;} \label{GeneralizedHeat-General}%
\end{align}
where $T\doteq1/(\partial S/\partial E)=1/\beta,$ $\beta\mathbf{A}%
\doteq(\partial S/\partial\boldsymbol{\xi})$ identifies the affinity
$\mathbf{A}$, and $\cdots$ refers to other elements in $\mathbf{X}$. Recalling
that for $\widetilde{\Sigma}$, $T=T_{0},P=P_{0},\cdots,\mathbf{A}_{0}=0$, we
have in general,
\end{subequations}
\begin{subequations}
\label{d-exch}%
\begin{align}
d_{\text{e}}W  &  =-d\widetilde{W}=P_{0}dV+\cdots,\label{deW}\\
d_{\text{e}}Q  &  =-d\widetilde{Q}=T_{0}d_{\text{e}}S. \label{deQ}%
\end{align}
We can identify the irreversible macrowork due to $\mathbf{F}_{\text{t}%
}^{\text{w}}$:
\end{subequations}
\begin{subequations}
\begin{equation}
d_{\text{i}}W=(P-P_{0})dV+\cdots+\mathbf{A}\cdot d\boldsymbol{\xi=}%
\mathbf{F}_{\text{t}}^{\text{w}}\cdot d\mathbf{Z}^{E}\geq0, \label{diW}%
\end{equation}
where
\begin{equation}
\mathbf{F}_{\text{t}}^{\text{w}}=\left\{  P-P_{0},\cdots,\mathbf{A}%
=\mathbf{A-A}_{0}\right\}  \label{Thermodynamic-Force/E}%
\end{equation}
is the \emph{thermodynamic force}, which also include the affinity
$\mathbf{A}$, driving the system towards EQ when $d_{\text{i}}W\rightarrow0$
as $\mathbf{F}_{\text{t}}^{\text{w}}\rightarrow0$. For irreversibility,
$d_{\text{i}}W>0$, which requires $\mathbf{F}_{\text{t}}^{\text{w}}$ to be
non-zero as asserted earlier in Sec. \ref{Sec-WhyThisStudy}. Each component
$(P-P_{0})dV,\cdots,A_{1}d\xi_{1},A_{2}d\xi_{2},\cdots$ of $\mathbf{F}%
_{\text{t}}^{\text{w}}\cdot d\mathbf{Z}^{E}$\ must be positive separately for irreversibility.

Using $d_{\text{e}}Q=T_{0}d_{\text{e}}S$ in Eq. (\ref{GeneralizedHeat-General}%
), we find
\end{subequations}
\begin{subequations}
\begin{equation}
d_{\text{i}}Q=\left\{
\begin{array}
[c]{c}%
(T-T_{0})d_{\text{e}}S+Td_{\text{i}}S\\
(T-T_{0})dS+T_{0}d_{\text{i}}S
\end{array}
\right.  \geq0. \label{diQ}%
\end{equation}
Using $d_{\text{i}}Q=d_{\text{i}}W$ and Eq. (\ref{diW}), we also obtain
\begin{equation}
d_{\text{i}}S=\left\{
\begin{array}
[c]{c}%
\left\{  (T_{0}-T)d_{\text{e}}S+\mathbf{F}_{\text{t}}^{\text{w}}\cdot
d\mathbf{Z}^{E}\right\}  /T\\
\left\{  (T_{0}-T)dS+\mathbf{F}_{\text{t}}^{\text{w}}\cdot d\mathbf{Z}%
^{E}\right\}  /T_{0}%
\end{array}
\right.  \geq0. \label{diS}%
\end{equation}
We see that $F_{\text{t}}^{\text{h}}=T_{0}-T$ can also be thought of as a
"thermodynamic force" due to thermal imbalance driving the system towards EQ
via heat transfer; it ensures that $F_{\text{t}}^{\text{h}}d_{\text{e}}S\geq
0$, in accordance with the second law. Thus, both contributions to
$d_{\text{i}}S$ are always nonnegative as expected. In the absence of any heat
exchange ($d_{\text{e}}S=0$) or for an isothermal system ($T=T_{0}$),$\ $we
have
\end{subequations}
\begin{equation}
d_{\text{i}}Q=Td_{\text{i}}S=d_{\text{i}}W,
\label{ClosedIsothermalSystem-InternalHeat}%
\end{equation}
where $d_{\text{i}}W$ is given by Eq. (\ref{diW}).

We can use
\begin{equation}
\mathbf{F}_{\text{t}}\doteq(F_{\text{t}}^{\text{h}},\mathbf{F}_{\text{t}%
}^{\text{w}}) \label{Generalized-Thermodynamic-Force}%
\end{equation}
as the \emph{generalized thermodynamic force}, which includes the thermal
imbalance $F_{\text{t}}^{\text{h}}$ and the work imbalance $\mathbf{F}%
_{\text{t}}^{\text{w}}$. It should be obvious that $F_{\text{t}}^{\text{h}}$
is meaningless for an isolated or an isothermal system, while $\mathbf{F}%
_{\text{t}}^{\text{w}}$ is meaningful for all NEQ systems, interacting or not.

\subsection{Internal Variables and the Isolated System}

The above formulation of MNEQT is perfectly suited for considering an isolated
system $\Sigma$ ($d_{\text{e}}W=d_{\text{e}}Q\equiv0$) so that Eqs.
(\ref{diQ-diW}-\ref{diE}) becomes the most important thermodynamic equality.
For an isolated system, $d\mathbf{X}=0$ so that $d_{\text{i}}W=\mathbf{A}\cdot
d\boldsymbol{\xi}$ as seen from Eq. (\ref{diW}).

\begin{theorem}
\label{Theorem-diW-diS-Isolated}The irreversible entropy generated within an
isolated system is still related to the dissipated work performed by the
internal variables.
\end{theorem}

\begin{proof}
As $E$ remains fixed for an isolated system ($dQ=Td_{\text{i}}S$), we have
from Eq. (\ref{FirstLaw-SI})%
\begin{equation}
Td_{\text{i}}S=d_{\text{i}}W=\mathbf{A}\cdot d\boldsymbol{\xi}\geq0
\label{EntropyDiff-Isolated}%
\end{equation}
in accordance with the second law.
\end{proof}

Note that the above equation, though it is identical to Eq.
(\ref{ClosedIsothermalSystem-InternalHeat}) in form, is very different in that
$d_{\text{i}}W$ here is simply $\mathbf{A}\cdot d\boldsymbol{\xi}$ and not the
full expression in Eq. (\ref{diW}). Same conclusion is also obtained when we
apply Eq. (\ref{diS}) to an isolated system.

The above theorem thus clarifies the unsettling fact about the significance of
dissipated macrowork that motivated this study; see also Eq.
(\ref{Isolated-Dissipated-Macrowork}). The dissipated macrowork $d_{\text{i}%
}W$ in an isolated system is performed by the internal variable
$\boldsymbol{\xi}$, and can be identified with $d_{\text{i}}S$ as noted in
Sec. \ref{Sec-WhyThisStudy}.

\begin{corollary}
\label{Corollary-IsolatedSystem}Neither the entropy can increase nor will
there be any dissipated work unless some internal variables are present in an
isolated system. If no internal variables are used to describe an isolated
system, then thermodynamics requires it to be in EQ.
\end{corollary}

\begin{proof}
The proof follows trivially from Eq. (\ref{EntropyDiff-Isolated}).
\end{proof}

\subsection{Cumulative Macroquantities in a Thermodynamic Process}

Let us consider a thermodynamic process $\mathcal{P}$ between two macrostates
$\mathfrak{M}_{\text{in}}\doteq\mathfrak{M}(T_{\text{in}},\mathbf{Z}%
_{\text{in}}^{E})$ and $\mathfrak{M}_{\text{fin}}\doteq\mathfrak{M}%
(T_{\text{fin}},\mathbf{Z}_{\text{fin}}^{E})$ at temperatures $T_{\text{in}}$
and $T_{\text{fin}}$, respectively. The system may be isolated and in IEQ so
its temperature is well defined. It may be very different from the temperature
$T_{0}$\ of the medium, if $\Sigma$ is not isolated. In each case, the
cumulative macroquantities of $\Sigma$\ are obtained by simple integration
along the process:%
\begin{subequations}
\begin{align}
\Delta_{\alpha}E  &  =%
{\textstyle\int\nolimits_{\mathcal{P}}}
d_{\alpha}E,\Delta_{\alpha}S=%
{\textstyle\int\nolimits_{\mathcal{P}}}
d_{\alpha}S,\label{CumulativeMacroquantity-E_S}\\
\Delta_{\alpha}W  &  =%
{\textstyle\int\nolimits_{\mathcal{P}}}
d_{\alpha}W,\Delta_{\alpha}Q=%
{\textstyle\int\nolimits_{\mathcal{P}}}
d_{\alpha}Q. \label{CumulativeMacroquantity-W_Q}%
\end{align}
Similar definitions also apply to $\widetilde{\Sigma}$ and $\Sigma_{0}$. Above
we have used the compact notation $d_{\alpha}=d,d_{\text{e}}$, and
$d_{\text{i}}$ to indicate various infinitesimal forms, which we can treat as
linear operators. They will be useful in the rest of the work.

We now consider an interacting system and determine $\Delta W$ between two EQ
macrostates $\mathfrak{M}_{\text{in,eq}}\doteq\mathfrak{M}(T_{0}%
,\mathbf{Z}_{\text{in}}^{E})$ and $\mathfrak{M}_{\text{fin,eq}}\doteq
\mathfrak{M}(T_{0},\mathbf{Z}_{\text{fin}}^{E})$\ at the same temperature. We
denote the corresponding process by $\mathcal{P}_{0}$, which\ may possibly be
irreversible. We recall that the Helmholtz free energy
\end{subequations}
\begin{subequations}
\begin{equation}
\overline{F}=E-T_{0}S \label{FreeEnergy}%
\end{equation}
(conventionally written as $F$ but we will use that for a SI free energy here)
is also the Helmholtz free energy of a NEQ\ system \cite{Gujrati-I} in a
canonical ensemble in a medium at temperature $T_{0}$; the temperature $T$ of
the system does not appear in $\overline{F}$ explicitly. It is this free
energy that follows the second law \cite{Gujrati-I} and not $F$, which is the
SI free energy
\begin{equation}
F=E-TS+\mathbf{A}\cdot\boldsymbol{\xi.} \label{SI-FreeEnergy}%
\end{equation}
It depends only on the system and is very different from $\overline{F}$. It
will be useful later. In terms of the difference $\Delta\overline{F}$ for the
two macrostates, we have $\Delta W=\Delta Q-\Delta E=\Delta Q-\Delta
\overline{F}-\Delta(T_{0}S)$. Thus,%

\end{subequations}
\begin{subequations}
\label{Del-W-alpha}%
\begin{align}
\Delta W  &  =%
{\textstyle\int\nolimits_{\mathcal{P}}}
[(T-T_{0})dS-SdT_{0}]-\Delta\overline{F},\label{Delta_W}\\
\Delta_{\text{e}}W  &  =-%
{\textstyle\int\nolimits_{\mathcal{P}}}
[T_{0}d_{\text{i}}S+SdT_{0}]-\Delta\overline{F},\label{Delta_W_e}\\
\Delta_{\text{i}}W  &  =%
{\textstyle\int\nolimits_{\mathcal{P}}}
[TdS-T_{0}d_{\text{e}}S] \label{Delta_W_i}%
\end{align}
The corresponding infinitesimal form is%
\end{subequations}
\begin{subequations}
\begin{align}
dW  &  =(T-T_{0})dS-SdT_{0}-d\overline{F},\label{d_W}\\
d_{\text{e}}W  &  =-[T_{0}d_{\text{i}}S+SdT_{0}]-d\overline{F},\label{d_W_e}\\
d_{\text{i}}W  &  =TdS-T_{0}d_{\text{e}}S. \label{d_W_i}%
\end{align}
If the medium is maintained at a fixed temperature during $\mathcal{P}$, we
must remove the $dT_{0}$ term above.

We see that if and only if $T=T_{0}=const$ over the entire process
$\mathcal{P}$ so that it is \emph{isothermal}, we have%
\end{subequations}
\begin{subequations}
\begin{equation}
\Delta W_{\text{isoth}}=-\Delta\overline{F},\Delta_{\text{e}}W_{\text{isoth}%
}=-T_{0}\Delta_{\text{i}}S-\Delta\overline{F}, \label{Iso-Delta_W0}%
\end{equation}
in terms of the Helmholtz free energy difference so that
\begin{equation}
\Delta_{\text{i}}W_{\text{isoth}}=T_{0}\Delta_{\text{i}}S.
\label{Iso-Delta_W_i}%
\end{equation}
This is a very strong requirement when $\widetilde{\Sigma}$ remains in
continuous contact with $\Sigma$, since it requires complete thermal
equilibrium at all times. In this case, $d_{\text{i}}Q=T_{0}d_{\text{i}}S$ and
$\Delta_{\text{i}}Q=T_{0}\Delta_{\text{i}}S$; see also Eq.
(\ref{ClosedIsothermalSystem-InternalHeat}).

We now consider an isolated system for which $\Delta E=0$ so that%
\end{subequations}
\begin{equation}
\Delta_{\text{i}}W=%
{\textstyle\int\nolimits_{\mathcal{P}}}
Td_{\text{i}}S=T_{0}\Delta_{\text{i}}S+%
{\textstyle\int\nolimits_{\mathcal{P}}}
[(T-T_{0})d_{\text{i}}S\geq0, \label{Isolated-Dissipated-Macrowork}%
\end{equation}
which is in accordance with Theorem \ref{Theorem-diW-diS-Isolated} and finally
explains the physical meaning of the irreversible macrowork, which was one of
the questions that had prompted this investigation. It should be clear that
the derivation is not restricted to a process between two EQ\ macrostates.

\subsection{A Simple Study Case \label{Sec-Example}}

As remarked earlier, we use a single internal variable $\xi$ in addition to
$V$ for $\mathbf{Z}^{E}$ for simplicity so that we have
\begin{subequations}
\begin{equation}
dW=PdV+Ad\xi,d_{\text{e}}W=P_{0}dV. \label{dW}%
\end{equation}
The dissipated work is%
\begin{equation}
d_{\text{i}}W=(P-P_{0})dV+Ad\xi\geq0; \label{diW-special}%
\end{equation}
in the absence of $\xi$. We also have%

\end{subequations}
\begin{equation}
Td_{\text{i}}S=(T-T_{0})d_{\text{e}}S+(P-P_{0})dV+Ad\xi\geq0,
\label{EntropyDiff-Isolated-SingleIV}%
\end{equation}
where we have used $d_{\text{i}}Q=d_{\text{i}}W$. If no exchange macrowork is
done, $dV=0$ and $Td_{\text{i}}S=(T-T_{0})d_{\text{e}}S+Ad\xi\geq0$. In the
absence of any exchange macroheat, we have $Td_{\text{i}}S=(P-P_{0}%
)dV+Ad\xi\geq0$. In an isolated system, we have
\begin{equation}
Td_{\text{i}}S=Ad\xi\geq0, \label{Isolated-IrreversibleEntropy}%
\end{equation}
which is a special case of the general result in Eq.
(\ref{EntropyDiff-Isolated}).

\section{The $\mu$NEQT\label{Sec-muNEQT}}

We will closely follow Refs. \cite{Gujrati-GeneralizedWork,Gujrati-LangevinEq}
to provide a brief pedagogical review of the $\mu$NEQT for the sake of
continuity and demonstrate its successful application to the free expansion,
which has not been attacked by any other approach so far. The main idea is to
cast any macroquantity in the MNEQT as a thermodynamic average, see Eq.
(\ref{Average-O}), over microstates. Then, we can identify corresponding
microquantities. Some care must be exercised to ensure their uniqueness as we
will see. These microstate microquantities can be used to identify the
contribution along a trajectory by simple integration.

\subsection{General Setup}

The theory was first presented in a very condensed form in Ref.
\cite{Gujrati-GeneralizedWork}. It was successfully applied
\cite{Gujrati-LangevinEq} to provide an alternative, but a much simpler,
approach (using deterministic microforces $\mathbf{F}_{k}^{E}$) to study
Brownian motion without the use of Langevin's stochastic noise term so it does
not require the use of the stochastic theory. The microforce responsible for
the Brownian motion is associated with the relative motion of the center of
mass of the Brownian particle with respect to the medium. A good description
of the salient features of the $\mu$NEQT is available there.

It follows from $\mathcal{H}(\left.  \mathbf{z}\right\vert \mathbf{Z}^{E})$
that $E_{k}$ also depend on $\mathbf{Z}^{E}$. A macroquantity $O$ (except the
temperature) in the MNEQT appear as a thermodynamic average $\left\langle
O\right\rangle $ over $\left\{  \mathfrak{m}_{k}\right\}  $:
\begin{subequations}
\begin{equation}
O=\left\langle O\right\rangle \doteq%
{\textstyle\sum\nolimits_{k}}
p_{k}O_{k}, \label{Average-O}%
\end{equation}
where $O_{k}$ is the value of $O$ associated with $\mathfrak{m}_{k}$. Thus,
\begin{equation}
E=%
{\textstyle\sum\nolimits_{k}}
p_{k}E_{k},P=%
{\textstyle\sum\nolimits_{k}}
p_{k}P_{k},S=%
{\textstyle\sum\nolimits_{k}}
p_{k}S_{k}, \label{Average-StateVariables}%
\end{equation}
etc. Here, $E_{k},P_{k}\doteq-\partial E_{k}/\partial V,S_{k}\doteq-\ln p_{k}%
$, etc. are the values that are only determined by $\mathfrak{m}_{k}$.
However, $p_{k}$ or $S_{k}$, although associated with\ $\mathfrak{m}_{k}$ is
not determined by it alone because of the constraint $%
{\textstyle\sum\nolimits_{k}}
p_{k}=1$, which makes it\ depend on the macrostate also.

From $E$, we have the change in the macroenergy
\end{subequations}
\begin{subequations}
\begin{equation}
dE=%
{\textstyle\sum\nolimits_{k}}
E_{k}dp_{k}+%
{\textstyle\sum\nolimits_{k}}
p_{k}dE_{k} \label{StateChange-Averages}%
\end{equation}
between two neighboring macrostates. We will use the compact notation%
\end{subequations}
\begin{align}
dE_{\text{h}}  &  \doteq%
{\textstyle\sum\nolimits_{k}}
E_{k}dp_{k}=\left\langle Ed\eta\right\rangle ,\nonumber\\
dE_{\text{w}}  &  \doteq%
{\textstyle\sum\nolimits_{k}}
p_{k}dE_{k}=\left\langle dE\right\rangle , \label{d_alpha-Oh}%
\end{align}
where we have introduced $\eta_{k}\doteq\ln p_{k}$ so that $d\eta_{k}%
=dp_{k}/p_{k}$.\ As $\left\{  E_{k}\right\}  $ does not change but $\left\{
p_{k}\right\}  $ changes in $dE_{\text{h}}$, it must depend on $dS$; see Eq.
(\ref{GeneralizedHeat-General}). As $\left\{  p_{k}\right\}  $ is not changed
in $dE_{\text{w}}$, it is evaluated at \emph{fixed entropy }$S$\emph{~}%
\cite{Gujrati-Heat-Work0,Gujrati-Entropy2}. Comparing $dE$ with the first law
in Eq. (\ref{FirstLaw-SI}), we identify%
\begin{subequations}
\begin{align}
dQ  &  =dE_{\text{h}}\doteq\left\langle Ed\eta\right\rangle
,\label{Macro-Heats}\\
dW  &  =-dE_{\text{w}}\doteq-\left\langle dE\right\rangle \label{MacroWorks}%
\end{align}
Thus, the generalized macroheat $dQ$ is a contribution proportional to $dS$
and the generalized macrowork $dW$ is an isentropic contribution as noted in
Sec. \ref{Sec-WhyThisStudy}. The identification of macroheat and macroworks
above also explains the choice of h for heat and w for work as suffix above
and superfix in $F_{\text{t}}^{\text{h}}$ and $\mathbf{F}_{\text{t}}%
^{\text{w}}$.

We can now identify the (generalized) microheat and microwork along
$\gamma_{k}$; they are given by%
\end{subequations}
\begin{equation}
dQ_{k}^{\prime}\doteq E_{k}d\eta_{k},dW_{k}=-dE_{k,} \label{Micro-Heat-Work}%
\end{equation}
respectively. The reason for the prime in $d_{\alpha}Q_{k}^{\prime}$\ will
become clear below. The second equation above is simply the previously derived
mechanical identity in Eq. (\ref{HamiltonianChange-StateSpace}). We summarize
this important result, which is not properly appreciated in the field
\cite{Gujrati-GeneralizedWork}, in the form of a

\begin{theorem}
\label{Theorem-dWk-dEk}The mechanical microwork $dW_{k}$ done by the system in
the $k$th microstate is the negative of the change $dE_{k}=(\partial
E_{k}/\partial\mathbf{Z}^{E})\cdot d\mathbf{Z}^{E}$:%
\begin{equation}
dW_{k}=-dE_{k} \label{microwork-microenergy}%
\end{equation}

\end{theorem}

\begin{proof}
See the derivation of Eq. (\ref{HamiltonianChange-StateSpace}).
\end{proof}

As shown in Eq. (\ref{HamiltonianChange-StateSpace}), the temporal evolution
of $\mathfrak{m}_{k}$ is due to $\mathbf{F}_{k}^{E}$, which changes its
microenergy $E_{k}$ but does not change $\mathfrak{m}_{k}$. The average over
all $\left\{  \mathfrak{m}_{k}\right\}  $ of $dW_{k}$ due to $\mathbf{F}%
_{k}^{E}$, see Eq. (\ref{Generalized force-work}), gives the generalized
macrowork $dW$, see Eq. (\ref{MacroWorks}), due to the macroforce
\[
\mathbf{F}^{E}\doteq\left\langle \mathbf{F}^{E}\right\rangle =T(\partial
S/\partial\mathbf{Z}^{E}),
\]
see Eq. (\ref{GeneralizedWork-General}), so that
\[
dW=\mathbf{F}^{E}\cdot d\mathbf{Z}^{E}%
\]
as expected. We can summarize the above result as the following theorem
because it plays a central role in the $\mu$NEQT.

In general, macroworks $d_{\alpha}W$ are thermodynamic averages of microworks
$d_{\alpha}W_{k}$:%
\begin{equation}
d_{\alpha}W=\left\langle d_{\alpha}W\right\rangle \doteq%
{\textstyle\sum\nolimits_{k}}
p_{k}d_{\alpha}W_{k}. \label{d_alpha-W-macro}%
\end{equation}
As it is easy to determine the mechanical microworks $d_{\alpha}W_{k}$, we can
extend the \emph{identity} $dE_{k}=-dW_{k}$ to introduce microenergies%
\begin{equation}
d_{\text{e}}E_{k}=-d_{\text{e}}W_{k},d_{\text{i}}E_{k}=-d_{\text{i}}W_{k};
\label{d_alpha-E-micro}%
\end{equation}
they define what is meant by $d_{\text{e}}E_{k}$ and $d_{\text{e}}E_{k}$.

Shifting $E_{k\text{ }}$by a constant does not affect $dW$ and $dQ$, showing
their unique nature. While $dW_{k}$ is also not affected by the shift, it does
affect $dQ_{k}^{\prime}$. Therefore, instead of using $dQ=dE_{\text{h}}$ to
identify $dQ_{k}$, we instead use the identity $dQ=TdS$ in Eq.
(\ref{GeneralizedHeat-General}) to identify $dQ_{k}$ as there is no ambiguity
in the definition of the statistical entropy $S$ in Eq.
(\ref{StatisticalEntropy}). We thus find that%
\begin{equation}
dQ_{k}\doteq T\left.  dS\right\vert _{k}\doteq-T(\eta_{k}+1)d\eta_{k},
\label{microheat}%
\end{equation}
where $d\left.  S\right\vert _{k}=-(\eta_{k}+1)d\eta_{k}$, not to be confused
with $dS_{k}=-d\eta_{k}$, is the microquantity corresponding to the
macroquantity $dS$%
\[
dS=-%
{\textstyle\sum\nolimits_{k}}
(\eta_{k}+1)dp_{k}\doteq%
{\textstyle\sum\nolimits_{k}}
p_{k}\left.  dS\right\vert _{k}.
\]
Similarly, we use Eqs. (\ref{deQ}) and (\ref{diQ}) to identify $d_{\text{e}%
}Q_{k}$ and $d_{\text{i}}Q_{k}$, respectively. We are not going to be directly
involved with microheats in our investigation here so we will not spend time
with them further; we will treat them in a separate publication.

\subsection{The Simple Case}

However, we need various microworks in this investigation as our focus is to
understand the dissipated work. Therefore, we give the results for them. For
the simple case $\mathbf{Z}^{E}=(V,\xi)$, compare with Sec. \ref{Sec-Example},
we have
\begin{subequations}
\begin{equation}
dW_{k}=P_{k}dV+A_{k}d\xi. \label{GeneralizedWork-Reduced}%
\end{equation}
From $d_{\text{e}}W=P_{0}dV$, we obtain $d_{\text{e}}W_{k}=P_{0}dV$, which
defines $d_{\text{e}}E_{k}=-P_{0}dV$. We also find
\begin{equation}
d_{\text{i}}W_{k}=(P_{k}-P_{0})dV+A_{k}d\xi\label{IrreversibleWork}%
\end{equation}
which identifies $d_{\text{i}}E_{k}=-(P_{k}-P_{0})dV-A_{k}d\xi$ and explains
how it becomes nonzero due to internal processes. For an isolated NEQ\ system,
we must $d_{\text{i}}W_{k}=A_{k}d\xi=dW_{k}$. For free expansion, we must set
$P_{0}=0$ so%
\begin{equation}
d_{\text{i}}W_{k}^{\text{free}}=P_{k}dV+A_{k}d\xi=-d_{\text{i}}E_{k}.
\label{FreeExpansion-Macrowork}%
\end{equation}

\subsection{Hamiltonian Trajectory and Microworks}

There have been several attempts to formulate microstate trajectory
thermodynamics \cite{Ritort,Esposito} based on utilizing the work fluctuation
theorem so they are not directly applicable to an isolated system. Our own
attempt that includes an isolated system was briefly outline in Ref.
\cite{Gujrati-GeneralizedWork} and elaborated recently in Ref.
\cite{Gujrati-JensenInequality}. Here, we briefly summarize it for continuity.
We will assume the medium $\widetilde{\Sigma}$ to consist of two
noninteracting media $\widetilde{\Sigma}_{\text{h}}$ that controls macroheat
exchange and $\widetilde{\Sigma}_{\text{w}}$ that controls macrowork exchange.
We are interested in a NEQ\ work process $\mathcal{P}_{0}$ as the system
evolves from one EQ macrostate to another by changing $\mathbf{X}^{E}$ from
$\mathbf{X}_{\text{in}}^{E}$ to $\mathbf{X}_{\text{fin}}^{E}$ by manipulating
the medium $\widetilde{\Sigma}_{\text{w}}$. It is usually the case that when
$\mathbf{X}^{E}=\mathbf{X}_{\text{fin}}^{E}$, the system is not yet in EQ so
the internal variables have not come to their EQ values. We denote this part
of $\mathcal{P}_{0}$ by $\mathcal{P}$. It takes a while during $\mathcal{\bar
{P}}$ ($\mathcal{P}_{0}=$ $\mathcal{P\cup\bar{P}}$) for the system to reach
the final EQ macrostate. We may allow the temperature $T_{0}$ of
$\widetilde{\Sigma}_{\text{h}}$ to change during $\mathcal{P}_{0}$, or
disconnect $\widetilde{\Sigma}_{\text{h}}$ from $\Sigma$\ during $\mathcal{P}%
$. In both cases, the temperature $T$ of the system in the final macrostate
$\mathfrak{M}_{\text{fin}}\doteq\mathfrak{M}(T_{\text{fin}},\mathbf{Z}%
_{\text{fin}}^{E})$\ may be different from that of the initial macrostate
$\mathfrak{M}_{\text{in}}\doteq\mathfrak{M}(T_{\text{in}},\mathbf{Z}%
_{\text{in}}^{E})$. While the microstate maintains its identity ($k$ does not
change) as shown in Eq. (\ref{Temporal mapping}), the microenergy $E_{k}$
changes during the entire evolution over $\mathcal{P}_{0}$ in accordance with
Eq. (\ref{HamiltonianChange-StateSpace}). Let us focus on $dW_{k}=-dE_{k}$
during $t$\ and $t+dt$ along $\gamma_{k}$. Its integral along $\gamma_{k}%
$\ determines the accumulated microwork $\Delta_{\alpha}W_{k}$
\end{subequations}
\begin{equation}
\Delta W_{k}\mathbf{\doteq}%
{\textstyle\int\nolimits_{\gamma_{k}}}
dW_{k}=-%
{\textstyle\int\nolimits_{\gamma_{k}}}
dE_{k}. \label{micro-works}%
\end{equation}
The integral is not affected by how $p_{k}$ changes during $\mathcal{P}_{0}$
so it is the same for all processes between $\mathbf{Z}_{\text{in}}^{E}$ and
$\mathbf{Z}_{\text{fin}}^{E}$. Thus, we can evaluate $\left\{  \Delta
W_{k}\right\}  $ for a single process such as an EQ process $\mathcal{P}%
_{0\text{,eq}}$ but can us it for every other possible process $\mathcal{P}%
_{0}(\mathbf{Z}_{\text{in}}^{E}\rightarrow\mathbf{Z}_{\text{fin}}^{E})$. On
the other hand, the accumulated macrowork $\Delta W$ over $\mathcal{P}_{0}$,
see Eq. (\ref{CumulativeMacroquantity-W_Q}), is affected by $p_{k}$, see Eq.
(\ref{d_alpha-W-macro}), so it is different for different processes
$\mathcal{P}_{0}(\mathbf{Z}_{\text{in}}^{E}\rightarrow\mathbf{Z}_{\text{fin}%
}^{E})$.

\begin{theorem}
\label{Theorem-Net-MicroWork} Let $\Delta\mathbf{Z}^{E}$ denote the change in
$\mathbf{Z}^{E}$ in the process $\mathcal{P}_{0}(\mathbf{Z}_{\text{in}}%
^{E}\rightarrow\mathbf{Z}_{\text{fin}}^{E})$. The cumulative microwork $\Delta
W_{k}=-\Delta E_{k}=E_{k}(\mathbf{Z}_{\text{in}}^{E})-E_{k}(\mathbf{Z}%
_{\text{in}}^{E}+\Delta\mathbf{Z}^{E})$ is the same for all processes,
including the reversible one, that undergo the same net change $\Delta
\mathbf{Z}^{E}$: $\mathbf{Z}_{\text{in}}^{E}\rightarrow\mathbf{Z}_{\text{fin}%
}^{E}=\mathbf{Z}_{\text{in}}^{E}+\Delta\mathbf{Z}^{E}$. However, the
cumulative macrowork $\Delta W$ depends on the process.
\end{theorem}

\begin{proof}
As $E_{k}$ is specific to the microstate $\mathfrak{m}_{k}$, the integral in
Eq. (\ref{micro-works}) is the required difference:%
\begin{equation}
\Delta W_{k}=-\Delta E_{k}=E_{k}(\mathbf{Z}_{\text{in}}^{E})-E_{k}%
(\mathbf{Z}_{\text{in}}^{E}+\Delta\mathbf{Z}^{E}).
\label{Accumulated-Microwork}%
\end{equation}
Let us consider various processes that occur when changing $\mathbf{Z}%
_{\text{in}}^{E}$ by $\Delta\mathbf{Z}^{E}$ to $\mathbf{Z}_{\text{in}}%
^{E}+\Delta\mathbf{Z}^{E}$, regardless of the process. As $E_{k}$\ is a
microproperty, the net difference $\Delta E_{k}$ is the same for all these
processes. As $\left\{  p_{k}\right\}  $ is different for different processes,
the macrowork $\Delta W$, see Eq. (\ref{CumulativeMacroquantity-W_Q}), is
usually different for different processes as expected.
\end{proof}

The theorem has far-reaching consequences. According to this, we can evaluate
$\left\{  \Delta W_{k}\right\}  $ for a single process such as an EQ process
$\mathcal{P}_{0\text{,eq}}$ but can use it for every other possible process
$\mathcal{P}_{0}(\mathbf{Z}_{\text{in}}^{E}\rightarrow\mathbf{Z}_{\text{fin}%
}^{E})$. On the other hand, the accumulated macrowork $\Delta W$ over
$\mathcal{P}_{0}$ is affected by $p_{k}$, see Eq. (\ref{d_alpha-W-macro}), so
it is different for different processes $\mathcal{P}_{0}(\mathbf{Z}%
_{\text{in}}^{E}\rightarrow\mathbf{Z}_{\text{fin}}^{E})$.

The identification in Eq. (\ref{microwork-microenergy}) or in
(\ref{Accumulated-Microwork}) is the most important feature of the $\mu$NEQT
that distinguishes it from current stochastic thermodynamic approaches, which
invariably identifies $d_{\text{e}}W_{k}$ with $-dE_{k}$ or $\Delta_{\text{e}%
}W_{k}$ with $-\Delta E_{k}\ $\cite{Gujrati-GeneralizedWork}; see also Remark
\ref{Remark-IncorrectIdentification-Energy}.

As the set $\left\{  \Delta W_{k}\right\}  $ is the same in all possible
processes with $\mathbf{Z}_{\text{in}}^{E}\rightarrow\mathbf{Z}_{\text{in}%
}^{E}+\Delta\mathbf{Z}^{E}$, we can introduce a \emph{random variable}
$\mathsf{W}$ such that it takes the value (outcome) $W_{k}=-E_{k}$, or a
random variable $d\mathsf{W}$ that takes the value $dW_{k}=-dE_{k}$ in
$\mathfrak{m}_{k}$ with probability $p_{k}$.

\begin{corollary}
\label{Corollary-SpontaneousEnergyDecrease}For a spontaneous process in an
isolated system,%
\begin{equation}
\Delta_{\text{i}}W_{k}\geq0,\Delta_{\text{i}}W\geq0.
\label{Spontaneous-InternalWork}%
\end{equation}

\end{corollary}

\begin{proof}
We see from the definition of generalized forces in Eq.
(\ref{Generalized force-work}) that $E_{k}$ acts like the potential energy of
a mechanical system. According to the principle of (potential) energy
minimization, these forces spontaneously cause the isolated mechanical system
to decrease $E_{k}$. Therefore, for a spontaneous process, $E_{k}%
(\mathbf{Z}_{\text{in}}^{E})\geq E_{k}(\mathbf{Z}_{\text{in}}^{E}%
+\Delta\mathbf{Z}^{E})$ in Eq. (\ref{Accumulated-Microwork}). As $\Delta
W_{k}=\Delta_{\text{i}}W_{k}$ for an isolated system, we have $\Delta
_{\text{i}}W_{k}\geq0$; hence $\Delta_{\text{i}}W\geq0$. This proves the corollary.
\end{proof}

\begin{remark}
\label{Remark-IncorrectIdentification-Energy}As $\Delta_{\text{e}}W_{k}$ is
associated with the MI macrowork $\Delta_{\text{e}}W$, it cannot be related to
the difference $dE_{k}$ of microstate energies that are SI quantities. This
clearly shows that the $\mu$NEQT is very different from current stochastic
thermodynamic approaches as noted above.
\end{remark}

The most important aspect of the Hamiltonian trajectory is the identity nature
of $\mathcal{T}$, which ensures that every initial microstate $\mathfrak{m}%
_{k},k\in\mathbb{N}$, is mapped onto itself at the end of $\mathcal{P}$,
although its energy will have changed. Thus, each $\gamma_{k}$ is unique and a
sum over its ensemble $\left\{  \gamma_{k}\right\}  $ is the same as the sum
over $k\in\mathbb{N}$. This is a major simplification of our approach, and
plays a major role in the rest of the study.

\section{Moment Generating Function \label{Sec-MGF0}}

\subsection{Trajectory Probability}

Let us consider an arbitrary process $\mathcal{P}$\ between $\mathfrak{M}%
_{\text{in}}\ $and $\mathfrak{M}_{\text{fin}}$. We consider two terminal
microstates $\mathfrak{m}_{k\text{,in}}\ $and $\mathfrak{m}_{k\text{,fin}}$
along $\mathcal{P}$\ and introduce the following trajectory probability
\cite{Gujrati-JensenInequality} between them for any system, interacting or
isolated, \ \
\begin{equation}
p_{\gamma_{k}}\equiv\frac{%
{\textstyle\int\nolimits_{\gamma_{k}}}
p_{k}dW_{k}}{%
{\textstyle\int\nolimits_{\gamma_{k}}}
dW_{k}}=%
{\textstyle\int\nolimits_{\gamma_{k}}}
p_{k}dw_{k}; \label{trajectory prob}%
\end{equation}
here, $dw_{k}\doteq dW_{k}/\Delta W_{k}$ and is independent of $p_{k}$. We can
generalize the above definition to introduce $p_{\gamma_{k},\alpha}$ by
replacing $dW_{k}$ by $d_{\alpha}W_{k}$; see Ref.
\cite{Gujrati-JensenInequality}.\ We see that
\begin{equation}%
{\textstyle\sum\nolimits_{k}}
p_{\gamma_{k}}\Delta W_{k}=%
{\textstyle\int\nolimits_{\gamma_{k}}}
dW=\Delta W, \label{Thermodynamic-probability}%
\end{equation}
the macrowork as introduced in Eq. (\ref{CumulativeMacroquantity-W_Q}). It is
clear that $%
{\textstyle\sum\nolimits_{k}}
p_{\gamma_{k}}=1$ as expected, and that $p_{\gamma_{k}}$ is nothing but the
(thermodynamic) probability of having a particular value $\Delta
W_{k}=E_{k,\text{in}}-E_{k,\text{fin}}$, determined only by $\Delta
\mathbf{Z}^{E}$ as seen from Theorem \ref{Theorem-Net-MicroWork}; here
$E_{k,\text{in}}=E_{k}(\mathbf{Z}_{\text{in}}^{E})$ and $E_{k,\text{fin}%
}=E_{k}(\mathbf{Z}_{\text{in}}^{E}+\Delta\mathbf{Z}^{E})$. It is clear that
$p_{\gamma_{k}}$ is the joint probability of $\mathfrak{m}_{k\text{,in}}\ $and
$\mathfrak{m}_{k\text{,fin}}$, which can be expressed in terms of the
conditional probability $p\left(  \left.  \mathfrak{m}_{k\text{,fin}%
}\right\vert \mathfrak{m}_{k\text{,in}}\right)  $:%
\begin{subequations}
\begin{equation}
p_{\gamma_{k}}=p(\mathfrak{m}_{k\text{,in}})p\left(  \left.  \mathfrak{m}%
_{k\text{,fin}}\right\vert \mathfrak{m}_{k\text{,in}}\right)  ,
\label{trajectory prob-cond}%
\end{equation}
where
\begin{equation}
p_{k\text{,in}}\doteq p(\mathfrak{m}_{k\text{,in}}) \label{Initial p_k}%
\end{equation}
is the probability of $\mathfrak{m}_{k\text{,in}}$. If $\mathfrak{M}%
_{\text{in}}$ is an EQ one at $T_{\text{in}}=T_{0}$, then the EQ probability
$p_{k\text{,in}}^{(0)}$ in the canonical ensemble is
\begin{equation}
p_{k\text{,in}}^{(0)}\doteq e^{\beta_{0}(\overline{F}_{\text{in}%
}-E_{k,\text{in}})}, \label{Initial p_k-eq}%
\end{equation}
where $\overline{F}_{\text{in}}$ is the initial EQ free energy and $\beta
_{0}=1/T_{0}$; for an isolated system, $T_{0}$ denotes its EQ temperature.

For a system in IEQ, with $\mathbf{Z}^{E}=(V,\boldsymbol{\xi})$, the
microstate probability looks very similar to the above Boltzmann probability%
\end{subequations}
\begin{equation}
p_{k}=\exp\left\{  \beta\lbrack\Phi-(E_{k}+P_{k}V+\mathbf{A}_{k}%
\cdot\boldsymbol{\xi})]\right\}  , \label{IEQ-p_k}%
\end{equation}
as given in Ref. \cite[Eq. (20) with $\mathbf{F}_{k\text{,BP}}\cdot\mathbf{R}$
replaced by $\mathbf{A}_{k}\cdot\boldsymbol{\xi}$]{Gujrati-LangevinEq}; here
$\Phi$ is the SI thermodynamic potential obtained by ensuring $%
{\textstyle\sum\nolimits_{k}}
p_{k}=1$.

Alternatively, we first determine $S$ using Eq. (\ref{StatisticalEntropy}).
This gives%
\begin{equation}
S=\beta\left(  E+PV+\mathbf{A}\cdot\boldsymbol{\xi-}\Phi\right)  .
\label{StatisticalEntropy-Explicit}%
\end{equation}
Using $p_{k}$ in a process involving IEQ macrostates, we determine
$p_{\gamma_{k}}$ above, which will be used in the rest of the paper. To obtain
the NEQ version of the canonical ensemble, the contribution from $V$ is
conventionally not included. In EQ, we must also remove the contribution from
$\boldsymbol{\xi}$ as it is not an independent variable anymore. This then
gives Eq. (\ref{Initial p_k-eq}).

We now determine $\Phi$ from Eq. (\ref{StatisticalEntropy-Explicit}) to
obtain
\[
\Phi=E-TS+PV+\mathbf{A}\cdot\boldsymbol{\xi.}%
\]
This is the SI potential in terms of system's macroquantities, and should not
be confused with the corresponding\ conventional thermodynamic potential
$\overline{\Phi}=E-T_{0}S+P_{0}V$ in terms of the fields of the medium. In a
NEQ canonical ensemble ($V$ fixed), $\Phi$ reduces to the SI free energy $F$
in Eq. (\ref{SI-FreeEnergy}), whereas the Helmholtz free energy is
$\overline{F}=E-T_{0}S$. In EQ, $\overline{\Phi}=$ $\Phi$ and $\overline{F}%
=F$; otherwise, they are different macroquantities.

\subsection{Moment Generating Function}

We introduce the moment generating function (MGF) for the random variable
$\Delta\mathsf{W}$, with outcomes $\Delta W_{k}$ over $\mathfrak{m}_{k}$,
along $\mathcal{P}$:
\begin{equation}
\mathcal{W}(\left.  \dot{\beta}\right\vert \Delta\mathsf{W})=\left\langle
e^{\dot{\beta}\Delta W}\right\rangle \doteq%
{\textstyle\sum\nolimits_{k}}
p_{\gamma_{k}}e^{\dot{\beta}\Delta W_{k}}, \label{Moment-GFn}%
\end{equation}
where $\dot{\beta}$ is some independent parameter that $p_{\gamma_{k}}$ and
$\Delta W_{k}$ do not depend upon, and the sum is over the ensemble of all
trajectories $\gamma_{k}$ originating at $\mathfrak{m}_{k\text{,in}}$. The
definition is valid for any system, interacting or isolated, and for any
terminal macrostates $\mathfrak{M}_{\text{in}}\ $and $\mathfrak{M}%
_{\text{fin}}$, neither of which has to be an EQ macrostate. Implicit in the
definition is the Hamiltonian characteristic of the process in Eq.
(\ref{Temporal mapping}) and the \emph{thermodynamic nature} of the trajectory
probability $p_{\gamma_{k}}$ in Eq. (\ref{Thermodynamic-probability}). The
latter makes $\mathcal{W}$ a thermodynamic function. The independent parameter
$\dot{\beta}$ should not be confused with any inverse temperature of $\Sigma$.

Various moments of $\Delta\mathsf{W}$\ are obtained by differentiating
$\mathcal{W}$ with respect to the parameter $\dot{\beta}$ and then setting
$\dot{\beta}=0$. For the first two moments, we have%
\begin{subequations}
\begin{align}
\left.  d\mathcal{W}/d\dot{\beta}\right\vert _{\dot{\beta}=0} &  =%
{\textstyle\sum\nolimits_{k}}
p_{\gamma_{k}}\Delta W_{k}=\Delta W,\label{MGF-FirstMoment}\\
\left.  d^{2}\mathcal{W}/d\dot{\beta}^{2}\right\vert _{\dot{\beta}=0} &  =%
{\textstyle\sum\nolimits_{k}}
p_{\gamma_{k}}\left(  \Delta W_{k}\right)  ^{2}=\Delta W^{2}%
>0;\label{MGF-Sec ondMoment}%
\end{align}
$\Delta W^{2}$ introduced above should not be confused with $(\Delta W)^{2}$.
Recalling Eq. (\ref{Thermodynamic-probability}), we see that $\Delta W$ is the
cumulative macrowork along the arbitrary process $\mathcal{P}$ for which there
is no sign restriction. Indeed, all moments derived from $\mathcal{W}$ are
thermodynamic in nature so they remain unchanged under any transformation such
as in Eq. (\ref{Energy-Infinite}) later that leaves $\mathcal{W}$ invariant.

The MGF, apart from yielding all the moments, is also quite useful in
establishing\ the following theorem for $\mathcal{P}_{0}$\ between two EQ
macrostates $\mathfrak{M}_{\text{in,eq}}$ and $\mathfrak{M}_{\text{fin,eq}}$.
\end{subequations}
\begin{theorem}
\label{Theorem-p_k-in-independence}For $\dot{\beta}=\beta_{0}=1/T_{\text{in}}%
$, the initial value of $\beta$ in $\mathfrak{M}_{\text{in,eq}}$, the MGF
becomes independence of the initial values of the probabilities $\left\{
p_{k\text{,in}}^{(0)}\right\}  $ .
\end{theorem}

\begin{proof}
Recalling Eq. (\ref{trajectory prob-cond}) and setting $p(\mathfrak{m}%
_{k\text{,in}})=p_{k\text{,in}}^{(0)}$, we see that the MGF reduces to a
function $\overline{\mathcal{W}}$ given by%
\begin{equation}
\overline{\mathcal{W}}(\left.  \beta_{0}\right\vert \Delta\mathsf{W}%
)=e^{\beta_{0}F_{\text{in}}}%
{\textstyle\sum\limits_{k}}
p_{\text{eq}}\left(  \left.  \mathfrak{m}_{k\text{,fin}}\right\vert
\mathfrak{m}_{k\text{,in}}\right)  e^{-\beta_{0}E_{k,\text{fin}}},
\label{Moment-GFn-1}%
\end{equation}
which clearly shows its independence from the initial probabilities $\left\{
p_{k\text{,in}}^{(0)}\right\}  $.
\end{proof}

The theorem proves extremely useful in the trick in Sec.
\ref{Sec-ClassicalExpansion} that is needed to overcome phase space volume
change. The bar on $\overline{\mathcal{W}}$ is a reminder that we are
considering a NEQ\ process $\mathcal{P}_{0}$\ between two EQ macrostates. We
also note that since $\beta_{0}$ is no longer an independent parameter,
$\overline{\mathcal{W}}$ is no longer a MGF.

\subsection{Some Special Cases}

The sum in $\overline{\mathcal{W}}$ turns into an EQ partition function if we
set $p\left(  \left.  \mathfrak{m}_{k\text{,fin}}\right\vert \mathfrak{m}%
_{k\text{,in}}\right)  =1,$ $\forall k$, a poor approximation, in which case
\begin{equation}
p_{k\text{,fin}}=p_{k\text{,in}}^{(0)},~p_{\gamma_{k}}(\mathcal{P}%
_{0})=p_{k\text{,in}}^{(0)},\forall k. \label{Prob-Approx}%
\end{equation}
This leads to a new function $\widehat{\mathcal{W}}$ (the caret on top is a
reminder of the choice in Eq. (\ref{Prob-Approx})),\ defined by%
\begin{equation}
\widehat{\mathcal{W}}(\left.  \beta_{0}\right\vert \Delta\mathsf{W}%
)=\widehat{\left\langle e^{\beta_{0}\Delta\mathsf{W}}\right\rangle }\doteq%
{\textstyle\sum\nolimits_{k}}
p_{k\text{,in}}^{(0)}e^{\beta_{0}\Delta W_{k}}, \label{MGF}%
\end{equation}
where the "average" denoted by $\widehat{\left\langle {}\right\rangle }$ is
with respect to $\left\{  p_{k\text{,in}}^{(0)}\right\}  $ as trajectory
probabilities. The new function simplifies to
\begin{equation}
\widehat{\mathcal{W}}(\left.  \beta_{0}\right\vert \left\{  \Delta
\mathsf{W}\right\}  )=e^{-\beta_{0}\Delta\overline{F}},
\label{GujratiRelation0}%
\end{equation}
where $\Delta\overline{F}\doteq\overline{F}_{\text{fin}}-\overline
{F}_{\text{in}}$. Comparing the "average" macrowork
\begin{equation}
\widehat{\left\langle \Delta\mathsf{W}\right\rangle }\doteq%
{\textstyle\sum\nolimits_{k}}
p_{k\text{,in}}^{(0)}\Delta W_{k}. \label{Reduced-Macrowork}%
\end{equation}
with Eq. (\ref{MGF-FirstMoment}), we conclude that%
\[
\Delta W\neq\widehat{\left\langle \Delta\mathsf{W}\right\rangle }.
\]
In general, the non-thermodynamic probability choice in $\widehat{\mathcal{W}%
}$ will not result in a thermodynamic macrowork $\Delta W$. It is clear that
the assumption $p_{\gamma_{k}}=p_{k\text{,in}}^{(0)}$ does not result in
thermodynamic average; see \cite{Gujrati-JensenInequality} for more details.
As $\widehat{\mathcal{W}}$ is a non-thermodynamic function, we will refer to
Eq. (\ref{GujratiRelation0}) as a \emph{mathematical identity} to
differentiate it from a thermodynamic identity. Despite this, it is an
interesting function in that it can be evaluated in a closed form as seen
above. The assumption of a constant $p_{\gamma_{k}},\forall k$ in Eq.
(\ref{Prob-Approx}) and the choice $\dot{\beta}=\beta_{0}$ has been
popularized by Jarzynski through his WFT to be discussed below.

The issue of \emph{non-thermodynamic }averaging was first raised by Cohen and
Mauzerall \cite{Cohen0,Cohen00} in the context of MFTs. But its most
significant consequence is about justifying the second law as discussed
earlier with respect to Eqs. (\ref{Jarzynski-Macrowork}%
-\ref{JarzynskiConjecture}). The issue has been discussed and settled only
recently in Ref. \cite{Gujrati-JensenInequality}. We will establish below that
the non-thermodynamic\emph{ identity} in Eq. (\ref{GujratiRelation0}) is
satisfied even for free expansion as expected as the $\mu$NEQT is perfectly
capable of describing isolated systems. For free expansion, $\Delta W_{k}$
does not identically vanish for all $k$ as we clearly see from Eq.
(\ref{FreeExpansion-Macrowork}). In this regard, our approach is very
different from the one used in the WFT, to which we now turn.

\subsection{The WFT\label{Sec-WFT}}

The non-thermodynamic function $\widehat{\mathcal{W}}$ is closely related to
the well-known Jarzynski's WFT \cite{Jarzynski} given below in Eq.
(\ref{JarzynskiRelation}). We first introduce the function used by Jarzynski
\
\begin{subequations}
\begin{equation}
\mathcal{W}_{\text{J}}(\left.  \beta_{0}\right\vert \Delta_{\text{e}%
}\mathsf{W})\doteq%
{\textstyle\sum\nolimits_{k}}
p_{k\text{,in}}^{(0)}e^{\beta_{0}\Delta_{\text{e}}W_{k}},
\label{JarzynskiRelation0}%
\end{equation}
which is obtained by replacing $\Delta\mathsf{W}$ by $\Delta_{\text{e}%
}\mathsf{W}$ in $\widehat{\mathcal{W}}(\left.  \beta_{0}\right\vert \left\{
\Delta\mathsf{W}\right\}  )$. Here, $\Delta_{\text{e}}\mathsf{W}$ is assumed
to be a random variable with outcomes $\Delta_{\text{e}}W_{k}$ over
$\mathfrak{m}_{k}$ with some probability, and the suffix "J" is a reminder for
Jarzynski's non-thermodynamic average in which the probability is replaced by
$p_{k\text{,in}}^{(0)}$. Jarzynski further assume, without proof, that
\end{subequations}
\begin{equation}
\Delta_{\text{e}}W_{k}\overset{\text{assump}}{=}\Delta W_{k}=-\Delta
E_{k},\forall k, \label{Jarzynski-WrongEquivalence}%
\end{equation}
as has been discussed recently for its validity \cite{Gujrati-GeneralizedWork}%
. This work-energy assumption is in addition to the non-thermodynamic average
conjecture in defining $\widehat{\mathcal{W}}$ in Eq. (\ref{MGF}), and
violates Theorem \ref{Theorem-Net-MicroWork}. With the two assumptions,
$\mathcal{W}_{\text{J}}$ becomes%
\begin{equation}
\mathcal{W}_{\text{J}}(\left.  \beta_{0}\right\vert \Delta_{\text{e}%
}\mathsf{W})=e^{-\beta_{0}\Delta\overline{F}}, \label{JarzynskiRelation}%
\end{equation}
which is the well-known WFT. Let us consider the following average
\begin{equation}
\Delta_{\text{e}}W_{\text{J}}\doteq\widehat{\left\langle \Delta_{\text{e}%
}\mathsf{W}\right\rangle }=%
{\textstyle\sum\nolimits_{k}}
p_{k\text{,in}}^{(0)}\Delta_{\text{e}}W_{k}. \label{Del_e-W_J}%
\end{equation}
From $\Delta_{\text{e}}W$ in Eq. (\ref{d_alpha-W-macro}), we see that in
general
\[
\Delta_{\text{e}}W_{\text{J}}\neq\Delta_{\text{e}}W=-R.
\]
This should be contrasted with the fundamental assumption of the WFT in Eq.
(\ref{JarzynskiConjecture}), which has already been questioned earlier
\cite{Cohen0,Gujrati-JensenInequality}. Thus, the conjecture in Eq.
(\ref{JarzynskiConjecture}) cannot be justified.\ Indeed, by considering an EQ
process in an interacting system, for which $\Delta_{\text{e}}W=\Delta
W=-\Delta\overline{F}$, it has been shown \ \cite[see Eq.(20) there]%
{Gujrati-JensenInequality} that%
\[
\Delta_{\text{e}}W_{\text{J}}<\Delta_{\text{e}}W,
\]
violating Eq. (\ref{JarzynskiConjecture}). For an isolated system undergoing
free expansion for which $\Delta_{\text{e}}W_{k}\equiv0$, the "mathematical
identity" in Eq. (\ref{JarzynskiRelation}) obviously fails \cite{Sung,Gross}.
Jarzynski \cite{Jarzynski-Gross} argues that such a system does not start in
EQ so the WFT should not apply there; however, see \cite{KestinNote} for counter-argument.

The "mathematical identity" in Eq. (\ref{JarzynskiRelation}) is evidently
different from the previous mathematical identity in Eq
(\ref{GujratiRelation0}) unless the above work-energy assumption in Eq.
(\ref{Jarzynski-WrongEquivalence}) is taken to be \emph{valid}; then the two
are the same. However, the WFT is considered a mathematical identity satisfied
for a class of NEQ processes by most workers in the field who have not
appreciated the implications of Eq (\ref{GujratiRelation0}). Let us determine
these implications within the $\mu$NEQT. It follows from Eq.
(\ref{Jarzynski-WrongEquivalence}) that\ $\Delta_{\text{i}}W_{k}=0,\forall k$;
see Theorem \ref{Theorem-Net-MicroWork}. Consequently, we must have
$\Delta_{\text{i}}W=0$ as noted in Sec. \ref{Sec-Prolog} so no irreversibility
is captured by the WFT as was first concluded a while back in Ref.
\cite{Gujrati-GeneralizedWork}. Using thermodynamic probabilities
$p_{\gamma_{k}}$ with the assumption in Eq. (\ref{Jarzynski-WrongEquivalence}%
), we find $\Delta W=\Delta_{\text{e}}W$ as a thermodynamic consequence. This
says nothing about $\widehat{\left\langle \Delta\mathsf{W}\right\rangle }$ or
$\Delta_{\text{e}}W_{\text{J}}$ that are not thermodynamic quantities.
However, it is commonly believed that $\Delta_{\text{e}}W_{\text{J}}%
=\Delta_{\text{e}}W$, which cannot be justified
\cite{Gujrati-JensenInequality}.

\section{The Free Expansion in The MNEQT\label{Sec-FreeExpansion-MNEQT}}

Our derivation of the identity in Eq. (\ref{GujratiRelation0}) is exact so it
should be valid for all processes including free expansion $\mathcal{P}_{0}$
as we will now show. In free expansion, there is no exchange of any kind so
$d=d_{i}$. This simplifies our notation as we do not need to use $d_{\text{i}%
}$ when referring to $\Sigma$. The gas $\Sigma$ expands freely in a vacuum
($\widetilde{\Sigma}$) from $V_{\text{in}}$, the volume of the left chamber,
to $V_{0}=V_{\text{fin}}$, the volume of $\Sigma_{0}$; the volume of the right
chamber is $V_{\text{fin}}-V_{\text{in}}$. The vacuum exerts no pressure
($\widetilde{P}=P_{\text{vacuum}}=0$). The left (L) and right (R) chambers are
initially separated by an impenetrable partition, shown by the solid partition
in Fig. \ref{Fig_Expansion}(a), to ensure that they are thermodynamically
independent regions, with all the $N$ particles of $\Sigma$ in the left
chamber, which are initially in an EQ macrostate $\mathfrak{M}_{\text{in,eq}}$
with entropy $S_{\text{in}}$. For ideal gas, we have%
\[
S_{\text{in}}=N\ln(eV_{\text{in}}/N);
\]
here, we are not including a temperature-dependent function \cite{Landau},
which does not play any role as we will be considering an isothermal free
expansion. The initial pressure and temperature of the gas prior to expansion
at time $t=0$ are $P_{\text{in}}$ and $T_{\text{in}}=T_{0}$, respectively,
that are related to $E_{0}=E_{\text{in}}$ and $V_{\text{in}}$ by its EQ
equation of state. As $\Sigma_{0}$ is isolated, the expansion occurs at
\emph{constant} energy $E_{0}$, which is also the energy of $\Sigma$.

It should be stated, which is also evident from Fig. \ref{Fig_Expansion}(b),
that while the removal of the partition is instantaneous, the actual process
of gas expanding in the right chamber is continuous and gradually fills it.
This is obviously a very complex internal process in a highly inhomogeneous
macrostate. As thus, it will require many internal variables to describe
different number of particles, different energies, different pressures,
different flow pattern which may be even chaotic, etc. in each of the
chambers. For example, we can divide the volume $V_{\text{fin}}$ into many
layers of volume parallel to the partition, each layer in equilibrium with
itself but need not be with others; see the example in Sec.
\ref{Sec-InternalVariables}. Here, we will simplify and take a single internal
variable $\xi$
\ \ \ \ \ \ \ \ \ \ \ \ \ \ \ \ \ \ \ \ \ \ \ \ \ \ \ \ \ \ \ \ \ \ \ \ \ \ \ \ \ \ \ \ \ \ \ \ \ \ \ \ \ \ \ \ \ \ \ \ \ \ \ \ \ \ \ \ \ \ \ \ \ \ \ \ \ \ \ \ \ \ \ \ \ \ \ \ \ \ \ \ \ \ \ \ \ \ \ \ \ \ \ \ \ \ \ \ \ \ \ \ \ \ \ \ \ \ \ \ \ \ \ \ \ \ \ \ \ \ \ \ \ \ \ \ \ \ \ \ \ \ \ \ \ \ \ \ \ \ \ \ \ \ \ \ \ \ \ \ \ \ \ \ \ \ \ \ \ \ \ \ \ \ \ \ \ \ \ \ \ \ \ \ \ \ \ \ \ \ \ \ \ \ \ \ \ \ \ \ \ \ \ \ \ \ \ \ \ \ \ \ \ \ \ \ \ \ \ \ \ \ \ \ \ \ \ \ \ \ \ \ \ \ \ \ \ \ \ \ \ \ \ \ \ \ \ \ \ \ \ \ \ \ \ \ \ \ \ \ \ \ \ \ \ \ \ \ \ \ \ \ \ \ \ \ \ \ \ \ \ \ \ \ \ \ \ \ \ \ \ \ \ \ \ \ \ \ \ \ \ \ \ \ \ \ \ \ \ \ \ \ \ \ \ \ \ \ \ \ \ \ \ \ \ \ \ \ \ \ \ \ \ \ \ \ \ \ \ \ \ \ \ \ \ \ \ \ \ \ \ \ \ \ \ \ \ \ \ \ \ \ \ \ \ \ \ \ \ \ \ \ \ \ \ \ \ \ \ \ \ \ \ \ \ \ \ \ \ \ \ \ \ \ \ \ \ \ \ \ \ \ \ \ \ \ \ \ \ \ \ \ \ \ \ \ \ \ \ \ \ \ \ \ \ \ \ \ \ \ \ \ \ \ \ \ \ \ \ \ \ \ \ \ \ \ \ \ \ \ \ \ \ \ \ \ \ \ \ \ \ \ \ \ \ \ \ \ \ \ \ \ \ \ \ \ \ \ \ \ \ \ \ \ \ \ \ \ \ \ \ \ \ \ \ \ \ \ \ \ \ \ \ \ \ \ \ \ \ \ \ \ \ \ \ \ \ \ \ \ \ \ \ \ \ \ \ \ \ \ \ \ \ \ \ \ \ \ \ \ \ \ \ \ \ \ \ \ \ \ \ \ \ \ \ \ \ \ \ \ \ \ \ \ \ \ \ \ \ \ \ \ \ \ \ \ \ \ \ \ \ \ \ \ \ \ \ \ \ \ \ \ \ \ \ \ \ \ \ \ \ \ \ \ \ \ \ \ \ \ \ \ \ \ \ \ \ \ \ \ \ \ \ \ \ \ \ \ \ \ \ \ \ \ \ \ \ \ \ \ \ \ \ \ \ \ \ \ \ \ \ \ \ \ \ \ \ \ \ \ \ \ \ \ \ \ \ \ \ \ \ \ \ \ \ \ \ \ \ \ \ \ \ \ \ \ \ \ \ \ \ \ \ \ \ \ \ \ \ \ \ \ \ \ \ \ \ \ \ \ \ \ \ \ \ \ \ \ \ \ \ \
\begin{equation}
\xi\doteq N_{\text{L}}-N_{\text{R}} \label{xi-free-expansion}%
\end{equation}
by considering only two layers to describe different numbers $N_{\text{L}%
}=(N+\xi)/2$ and $N_{\text{R}}=(N-\xi)/2$ of particles as a function of time.
At each instant, we imagine a front of the expanding gas shown by the solid
vertical line in Fig. \ref{Fig_Expansion}(b) containing all the particles to
its left. We denote this volume by a time-dependent $V=V(t)$ to the right of
which exists a vacuum. This means that at each instant when there is a vacuum
to the right of this front, the gas is expanding against zero pressure so that
$d_{\text{e}}W=0$. Since we have a NEQ expansion, $dW>0$. As $V(t)$ cannot be
controlled externally, it also represents an internal variable. The two
internal variables $\xi(t)$\ and $V(t)$\ allow us to distinguish between
$\mathcal{P}$ and $\mathcal{\bar{P}}$ as we will see below. We assume that the
expansion is isothermal so there is no additional internal variable associated
with temperature variation. As $dQ=dW\neq0$, the expansion is irreversible so
the entropy continues to change (increase).

At $t=0$, the partition is suddenly removed, shown by the broken partition in
Fig. \ref{Fig_Expansion}(b) and the gas expands freely to the final volume
$V(t^{\prime})=V_{\text{fin}}$ at time $t^{\prime}<\tau_{\text{eq}}$ during
$\mathcal{P}$. At $t^{\prime}$, the free expansion stops but there is no
reason a priori for $\xi=0$ so the gas is still inhomogeneous ($\xi\neq0$).
This is in a NEQ macrostate until $\xi$ achieves its EQ value $\xi=0$ during
$\mathcal{\bar{P}}$, at the end of which at $t=\tau_{\text{eq}}$ the gas
eventually comes into $\mathfrak{M}_{\text{fin,eq}}$ isoenergetically. We
briefly review this expansion in the MNEQT \cite{Gujrati-Heat-Work}.

We work in the state space $\mathfrak{S}$. Using Eq.
(\ref{GeneralizedHeat-General}), we have%
\begin{subequations}
\begin{equation}
dS(t)=dW(t)/T(t). \label{dS-dW-Free Expansion}%
\end{equation}
Setting $P_{0}=0$ in Eq. (\ref{diW}), we have
\begin{equation}
dW(t)=\left\{
\begin{array}
[c]{c}%
P(t)dV(t)+A(t)d\xi(t)\text{ for }t<t^{\prime}<\tau_{\text{eq}},\\
A(t)d\xi(t)\text{ for }t^{\prime}<t\leq\tau_{\text{eq}};
\end{array}
\right.  \label{dW-Free Expansion}%
\end{equation}
here, we have used the fact that $V(t)$ does not change for $\tau^{\prime
}<t\leq\tau_{\text{eq}}$. Thus,
\end{subequations}
\begin{align*}
\Delta S  &  =\int_{\mathcal{P}_{0}}\frac{dW(t)}{T(t)}>0,\\
\Delta Q  &  =\int_{\mathcal{P}_{0}}dW(t)=\Delta W>0;
\end{align*}
the last equation is the fundamental identity in Eq. (\ref{diQ-diW}). The
irreversible entropy change $\Delta S$\ from EQ macrostate from $\mathfrak{M}%
_{\text{in,eq}}$ to $\mathfrak{M}_{\text{fin,eq}}$ during $\mathcal{P}_{0}%
$\ is the EQ entropy change $\Delta_{\text{i}}S$ is%
\begin{equation}
\Delta S\equiv S_{\text{fin}}-S_{\text{in}}, \label{Entropy-Change-Eq}%
\end{equation}
and can be directly obtained if the EQ entropy $S(E,V)$ is known. The above
analysis is also valid for any arbitrary free expansion process $\mathcal{P}$;
we must carry out the integration over $\mathcal{P}$\ above. We can evaluate
$\Delta_{\text{i}}S$ by using $S$ from Eq. (\ref{StatisticalEntropy-Explicit}).

The above exercise allows us to identify $\Delta W$ as the dissipated work
over the entire process even if the process. We have $\Delta W=\Delta
_{\text{i}}W=T_{0}\Delta_{\text{i}}S$ for an isothermal process $T(t)=T_{0}$;
see also \ Eq. (\ref{Iso-Delta_W_i}). However, with the inclusion of the
internal variables above, we are also able to determine $d_{\text{i}%
}S=dS=dW/T$ using Eq. (\ref{dW-Free Expansion}) for any infinitesimal segment
$\delta\mathcal{P}$ of the process in the MNEQT that was one of our goals.
Thus, $d_{\text{i}}W=Td_{\text{i}}S$ is the infinitesimal "dissipated work"
over $\delta\mathcal{P}$ but the relationship contains $T$ and not $T_{0}$ as
it must; see Theorem \ref{Theorem-diW-diS-Isolated}.

Let us consider an ideal gas for which $V_{\text{fin}}=2V_{\text{in}}$ so that
$\Delta_{\text{i}}S=N\ln2$, a well-known result \cite{Prigogine}. Here, we
provide a more general result for the entropy for $t\leq t^{\prime}$, which
can be trivially determined:%
\[
S(t,\xi)=N_{\text{L}}\ln(eV_{\text{in}}/N_{\text{L}})+N_{\text{R}}%
\ln(eV^{\prime}/N_{\text{R}}),
\]
with $\xi>0$; here $V^{\prime}=V-V_{\text{in}}$. Thus, for arbitrary $\xi$, we
have $\Delta_{\text{i}}S(t,\xi)=S(t,\xi)-S_{\text{in}}$. At EQ, not only
$V_{\text{fin}}^{\prime}=V_{\text{fin}}-V_{\text{in}}$, but also $\xi=0$ so
the EQ entropy is given by
\[
S_{\text{fin}}=(N/2)\ln(2eV_{\text{in}}/N)+(N/2)\ln(2eV_{\text{fin}}^{\prime
}/N),
\]
which is consistent with $\Delta_{\text{i}}S=N\ln2$, as expected. We can also
take the initial macrostate to be not an EQ one in $\mathcal{P}$ by using one
or more additional internal variables. Thus, the approach is very general.

\section{The Free Expansion in The $\mu$NEQT\label{Sec-FreeExpansion}}

\subsection{Quantum Free Expansion \label{Sec-QuantumExpansion}}

The expansion/contraction of a one-dimensional quantum \emph{ideal} gas with
moving walls has been treated in many different ways
\cite{Rice,Fojon,Martino,Cooney} but none deal with sudden expansion. The
latter, however, has been studied
\cite{Bender,Gujrati-QuantumHeat,Gujrati-JensenInequality} quantum
mechanically (without any $\xi$) as a particle in an isolated box $\Sigma_{0}%
$\ of length $L_{\text{fin}}$, which we restrict to $2L_{\text{in}}$ here,
with rigid, insulating walls. We briefly revisit this study and expand on it
by introducing a $\xi$ to set the stage for the classical expansion using the
$\mu$NEQT\ in the following section. We will follow Ref.
\cite{Gujrati-JensenInequality} closely.

We make the very simplifying assumptions in the previous section to introduce
$\xi$. At time $t=0$, all the $N$ particles (or their wavefunctions) are
confined in EQ in the left chamber of length $L_{\text{in}}$ so that
$N_{\text{L}}=N$ initially. We can think of an intermediate length
$L_{\text{fin}}\geq L(t)>L_{\text{in}}$, in analogy with $V(t)$ in the
previous section, so that $N_{\text{R}}=N-N_{\text{L}}$ particles are
simultaneously confined in the intermediate chamber of size $L(t)$, while
$N_{\text{L}}$ particles are still confined in the left chamber for all $t>0$
as we did in the previous section. Eventually, at $t=\tau_{\text{eq}}$, all
the $N_{\text{R}}=N$ particles are confined in the larger chamber of size
$L_{\text{fin}}$ so that there are no particles in the initial chamber. We let
$\xi=N_{\text{L}}$, which gradually decreases from $\xi=N$ to $\xi=0$. Note
that this definition is different from the previous section but we make this
choice for the sake of simplicity.\ At some intermediate time $\tau^{\prime
}<\tau_{\text{eq}}$ that identifies $\mathcal{P}$, $L(t)=L_{\text{fin}}$, but
$N_{\text{R}}$ is still not equal to $N$ ($\xi\neq0$). We then follow its
equilibration during $\mathcal{\bar{P}}$ as the gas come to EQ in the larger
chamber at the end of $\mathcal{P}_{0}$ when $\xi=0$. Again, there are two
internal variables $L$ and $\xi$. The expansion is isoenergetic at each
instant. As we will see below, this means that it is also isothermal. However,
$dQ=dW\neq0$ ensuring a irreversible process so the microstate probabilities
continue to change.

Since we are dealing with an ideal gas, we can focus on a single particle
whose energy levels are in appropriate units $E_{k}=k^{2}/l^{2}$, where $l$ is
the length of the chamber confining it. The single-particle partition function
for arbitrary $l$ and inverse temperature $\beta=1/T$\ is given by
\[
Z(\beta,l)=%
{\textstyle\sum\nolimits_{k}}
e^{-\beta E_{k}(l)},
\]
from which we find that the single particle free energy is $\overline
{F}=-(T/2)\ln(\pi Tl^{2}/4)$ and the average single particle energy is
$E=1/2\beta$, which depends only on $\beta$ but not on $l$. Assuming that the
gas is in IEQ so that the particles in each of the two chambers are in EQ (see
the second example in Sec. \ref{Sec-InternalVariables}) at inverse
temperatures $\beta_{\text{L}}$ and $\beta$, we find that the $N$-particle
partition function is given by%
\[
Z_{N}(\beta_{\text{L}},\beta)=\left[  Z(\beta_{\text{L}},L_{\text{in}%
})\right]  ^{\xi}\left[  Z(\beta,L)\right]  ^{N-\xi}%
\]
so that the average energy is $E_{N}(\beta_{\text{L}},\beta,L_{\text{in}%
},L,\xi)=\xi/2\beta_{\text{L}}+(N-\xi)/2\beta$. As this must equal
$N/2\beta_{0}$ for all values of $L$ and $\xi$, it is clear that
$\beta_{\text{L}}=\beta=\beta_{0}$, which proves the above assertion of an
isothermal free expansion at $T_{0}$.

To determine $\Delta W_{k}$, we merely have to determine the microenergy
change $\Delta E_{k}=E_{k\text{,fin}}-E_{k\text{,in}}$. It is trivially seen
that Eq. (\ref{GujratiRelation0}) is satisfied as was reported earlier
\cite{Gujrati-GeneralizedWork,Gujrati-JensenInequality}.

Below we will show that the quantum calculation here deals with an
irreversible $\mathcal{P}_{0}$. The single-particle energy change $\Delta
E_{k}$ is
\[
\Delta E_{k}=k^{2}(1/L^{2}-1/L_{\text{in}}^{2})<0,L>L_{\text{in}}.
\]
The micropressure%
\begin{equation}
P_{k}=-\partial E_{k}/\partial L=2E_{k}/L\neq0 \label{P_k-FreeExpansion}%
\end{equation}
determines the microwork%
\begin{equation}
\Delta W_{k}=%
{\displaystyle\int\nolimits_{L_{\text{in}}}^{L_{\text{fin}}}}
P_{k}dL>0. \label{DW-FreeExpansion}%
\end{equation}
It is easy to see that this microwork is precisely equal to \ $(-\Delta
E_{k})$ as expected; see Eq. (\ref{micro-works}). It is also evident from Eq.
(\ref{P_k-FreeExpansion}) that for each $L$ between $L_{\text{in}}$ and
$L_{\text{fin}}$,
\[
P=%
{\textstyle\sum\nolimits_{k}}
p_{k}P_{k}=2E/L\neq0,
\]
We can use this average pressure to calculate the thermodynamic macrowork
\[
\Delta W=%
{\displaystyle\int\nolimits_{L_{\text{in}}}^{L_{\text{fin}}}}
PdL=2%
{\textstyle\sum\nolimits_{k}}
{\displaystyle\int\nolimits_{L_{\text{in}}}^{L_{\text{fin}}}}
p_{k}E_{k}dL/L\neq0.
\]
as expected. As $\Delta E=0$, this means that the irreversible macroheat and
macrowork are $\Delta Q=\Delta W>0$. This establishes that the expansion we
are studying is \emph{irreversible}.

We now turn to the entire system in which the work is done by $N_{\text{R}}$
particles. We need to think of the microstate index $k$ as an $N$-component
vector $\mathbf{k}=\left\{  k_{i}\right\}  $ denoting the indices for the
single-particle microstates. For a given $\xi$, we have $\Delta W_{\mathbf{k}%
}(L,\xi)=-%
{\textstyle\sum\nolimits_{i}}
\Delta E_{k_{i}}$, where $i$ runs over the $N_{\text{R}}$\ particles. We can
compute the macrowork, which turns out to be $\Delta W_{N}(\xi)=(N-\xi)\Delta
W>0$. The corresponding change in the free energy is
\begin{align*}
\Delta\overline{F}_{N}(L,\xi)  &  =(N-\xi)[\overline{F}(\beta_{0}%
,L)-\overline{F}(\beta_{0},L_{\text{in}})]\\
&  =-\Delta W_{N}(\xi),
\end{align*}
which is consistent with Eq. (\ref{Iso-Delta_W0}) for an isolated system for
any $\xi$.

At the end of $\mathcal{P}_{0}$, $\Delta W_{N}(0)=N\Delta W>0$, and
$\Delta\overline{F}_{N}(0)=N[\overline{F}(\beta_{0},L_{\text{fin}}%
)-\overline{F}(\beta_{0},L_{\text{in}})]$. We can now set up the MGF
$\mathcal{W}(\left.  \beta\right\vert \left\{  \Delta\mathsf{W}\right\}  )$
for any $L$ and $\xi$ so that we can compute all the moments. However, we will
only consider the entire process $\mathcal{P}_{0}$ so that $\xi=0$ at the end.
For the first moment in Eq. (\ref{MGF-FirstMoment}) we find that for the
isothermal expansion
\begin{equation}
\Delta W_{N}=-\Delta\overline{F}_{N}=T_{0}\Delta_{\text{i}}S_{N}>0,
\label{DW-DF-DiS}%
\end{equation}
after using Eq. (\ref{Iso-Delta_W_i}). The same result is also obtained from
the classical isothermal expansion; see Eq. (\ref{dS-dW-Free Expansion}). All
this is in accordance with Theorem \ref{Theorem-diW-diS-Isolated} in the
MNEQT, as expected.

For the discussion below, we suppress $N$ as the subscript for simplicity. The
benefit of using the $\mu$NEQT is that we get a much better perspective of the
dissipation in the free expansion. In terms of the SI microwork,
$\Delta_{\text{i}}W_{k}\neq0,\forall k$ even though the SI microwork
$\Delta_{\text{e}}W_{k}\equiv0,\forall k$. However, what is more revealing
about the free expansion is that $\Delta W_{k}>0$ for each Hamiltonian
trajectory $\gamma_{k}$ as seen from Eq. (\ref{DW-FreeExpansion}), which is in
accordance with Corollary \ref{Corollary-SpontaneousEnergyDecrease}. This is
not true in a general process. For example, $\Delta_{\text{i}}W=0$ in a
reversible process even though $\Delta_{\text{i}}W_{k}\neq0,\forall k$. In
fact, $\Delta_{\text{i}}W_{k}$ must be of either sign to ensure a vanishing
average. \ 

\subsection{Classical Free Expansion\label{Sec-ClassicalExpansion}}

The change in the phase space volume in classical statistical mechanics
destroys the required unique mapping in Eq. (\ref{Temporal mapping}) as
discussed in Sec. \ref{Sec-VoluemChange}, and causes a problem with the use of
the Hamiltonian trajectories that were used in deriving the $\mu$NEQT and
introducing the MGF $\mathcal{W}$. Their use for classical expansion will
requires some modification, which we describe below. To introduce the required
modification as simply as possible, we will first consider $V_{\text{fin}%
}=2V_{\text{in}}$ that results in doubling the volume after expansion; see
Fig. \ref{Fig_Expansion}(a). Later, we will generalize to any arbitrary
expansion/contraction. We will use the notation of Sec.
\ref{Sec-FreeExpansion-MNEQT}, and restrict ourselves to $\mathbf{Z}%
^{E}=\left(  V_{\text{in}},V_{\text{fin}},V,\xi\right)  $; see Eq.
(\ref{xi-free-expansion}) for $\xi$ and Fig. \ref{Fig_PhasePoint-Evolution}.

Let $\delta\mathbf{z}(t)\doteq\delta\mathbf{z}\left(  V,\xi\right)  $ denote a
microstate at some time $t$ with work variables $V\doteq V(t)$ and $\xi
\doteq\xi(t)$. Let $E_{k}(t)\doteq E(\delta\mathbf{z}(t))$ denote the
instantaneous microenergy of $\delta\mathbf{z}(t)$. Let $\delta\mathbf{z}$ be
some initial microstate at $t=0$ with $N_{\text{L}}=N,N_{\text{R}%
}=0,V=V_{\text{in}}$, and $\xi=\xi_{\text{in}}=N$. We denote the number of
microstates in the initial phase space that is denoted by the interior of the
solid red ellipse $\boldsymbol{\Gamma}_{\text{in}}=\boldsymbol{\Gamma
}_{\text{in}}(\mathbf{Z}_{\text{in}}^{E})$ on the left by $\mathcal{N}$. The
final phase space is shown by the solid red ellipse $\boldsymbol{\Gamma
}_{\text{fin}}^{\prime}=\boldsymbol{\Gamma}_{\text{fin}}^{\prime}%
(\mathbf{Z}_{\text{fin}}^{E})$ on the right, which contains twice as many
($2\mathcal{N}$) microstates as are in $\boldsymbol{\Gamma}_{\text{in}}$. We
will assume that both the initial and the final macrostates ($\mathcal{M}%
_{\text{in,eq}}$ and $\mathcal{M}_{\text{fin,eq}}$) are in EQ for which $V$
and $\xi$ are no longer independent state variables. Therefore, we will not
use $V$ and $\xi$\ in $\mathbf{Z}^{E}$ for the two phase spaces at $t=0$ and
$t=\tau_{\text{eq}}$. We set $P_{0}=0$ in Eq. (\ref{IrreversibleWork}) to
obtain the microwork for $t>0$
\[
dW_{k}=P_{k}dV+A_{k}d\xi,
\]
which is nonzero, while $d_{\text{e}}W_{k}=0$. The EQ gas at $t=0$ has
microstate probabilities
\ \ \ \ \ \ \ \ \ \ \ \ \ \ \ \ \ \ \ \ \ \ \ \ \ \ \ \ \ \ \ \ \ \ \ \ \ \ \ \ \ \ \
\begin{subequations}
\begin{align}
p_{\text{in}}(\delta\mathbf{z})  &  =e^{-\beta_{0}E(\delta\mathbf{z}%
)}/Z_{\text{in}}(\beta_{0},V_{\text{in}})\geq
0,\label{probability-dist-initial}\\
Z_{\text{in}}(\beta_{0},V_{\text{in}})  &  \doteq%
{\textstyle\sum\limits_{\delta\mathbf{z}\in\boldsymbol{\Gamma}_{\text{in}}}}
e^{-\beta_{0}E(\delta\mathbf{z})}. \label{PF-initial}%
\end{align}

Under $\mathcal{T}$, $\delta\mathbf{z}(V_{\text{in}})\in\boldsymbol{\Gamma
}_{\text{in}}$ maps onto its image (see the green double arrow $\gamma
_{\text{in}}$) $\delta\mathbf{z}^{\prime}(V_{\text{fin}})\in\boldsymbol{\Gamma
}_{\text{in}}^{\prime}=\boldsymbol{\Gamma}_{\text{in}}^{\prime}(V_{\text{fin}%
})$, where $\boldsymbol{\Gamma}_{\text{in}}^{\prime}$ is shown schematically
by the broken red ellipse on the right. Thus, $\boldsymbol{\Gamma}_{\text{in}%
}^{\prime}$ contains $\mathcal{N}$ microstates. We then pick a microstate
$\delta\mathbf{\zeta}^{\prime}(V_{\text{fin}})\in\boldsymbol{\Gamma
}_{\text{diff}}^{\prime}(V_{\text{fin}})\boldsymbol{\doteq\Gamma}_{\text{fin}%
}^{\prime}\backslash\boldsymbol{\Gamma}_{\text{in}}^{\prime}$; here,
$\boldsymbol{\Gamma}_{\text{diff}}^{\prime}$ denotes $\boldsymbol{\Gamma
}_{\text{fin}}^{\prime}$ from which $\boldsymbol{\Gamma}_{\text{in}}^{\prime}%
$\ has been removed so it also contains $\mathcal{N}$ microstates. Because of
the uniqueness of $\gamma_{\text{fin}}$, $\delta\mathbf{\zeta}^{\prime
}(V_{\text{fin}})\mathbf{\notin}\boldsymbol{\Gamma}_{\text{in}}^{\prime}$ is
the image of a microstate $\delta\mathbf{\zeta}(V_{\text{in}})\in
\boldsymbol{\Gamma}_{\text{diff}}\boldsymbol{\doteq\Gamma}_{\text{fin}%
}\backslash\boldsymbol{\Gamma}_{\text{in}}$, where $\boldsymbol{\Gamma
}_{\text{fin}}$ is shown by the broken ellipse on the left from which
$\boldsymbol{\Gamma}_{\text{in}}$ has been taken out to obtain
$\boldsymbol{\Gamma}_{\text{diff}}$. Again, $\boldsymbol{\Gamma}_{\text{diff}%
}$ contains $\mathcal{N}$ microstates. To find $\delta\mathbf{\zeta}$, we
follow the $inverse$ of $\gamma_{\text{fin}}$ along which $\mathbf{Z}%
_{\text{fin}}^{E}\rightarrow\mathbf{Z}_{\text{in}}^{E}$ ($V_{\text{fin}%
}\rightarrow V_{\text{in}}$); this is shown by the left arrow on
$\gamma_{\text{fin}}$ in accordance with the reversibility of $\mathcal{T}$.
The same number of microstates in $\boldsymbol{\Gamma}_{\text{diff}}$ and
$\boldsymbol{\Gamma}_{\text{diff}}^{\prime}$ ensures the uniqueness of
$\gamma_{\text{fin}}$, a prerequisite for Hamiltonian trajectories.
Physically, the $\mathcal{N}$ microstates in$\ \boldsymbol{\Gamma
}_{\text{diff}}$ corresponds to as if there are $N$ particles in the right
chamber in Fig. \ref{Fig_Expansion}(a), with the partition intact so that each
particle is confined in the volume $V$ of the right chamber only. But note
that there are no particles in the right chamber at $t=0$.

The situation was very simple for the uniqueness of $\gamma_{\text{fin}}$
because of the choice $V_{\text{fin}}=2V_{\text{in}}$: $\boldsymbol{\Gamma
}_{\text{diff}}^{\prime}$ and $\boldsymbol{\Gamma}_{\text{diff}}$ have the
same number of $\mathcal{N}$ microstates. This will not be the case if
$V_{\text{in}}<V_{\text{fin}}<2V_{\text{in}}$, since the number of microstates
$\delta\mathbf{\zeta}^{\prime}(V_{\text{fin}})$ in $\boldsymbol{\Gamma
}_{\text{diff}}^{\prime}$ is now $\mathcal{N}^{\prime}<\mathcal{N}$, while
$\boldsymbol{\Gamma}_{\text{diff}}$ has $\mathcal{N}$ microstates
$\delta\mathbf{\zeta}(V_{\text{in}})$. To identify the set $\left\{
\delta\mathbf{\zeta}(V_{\text{in}})\right\}  \subset\boldsymbol{\Gamma
}_{\text{diff}}$ of theses $\mathcal{N}^{\prime}$\ microstates, we follow the
inverse of $\gamma_{\text{fin}}$ for each $\delta\mathbf{\zeta}^{\prime
}(V_{\text{fin}})\in\boldsymbol{\Gamma}_{\text{diff}}^{\prime}(V_{\text{fin}%
})$ to identify the image $\delta\mathbf{\zeta}(V_{\text{in}})$ as described
above. From now on, we will use $\boldsymbol{\Gamma}_{\text{diff}}(\left.
V_{\text{in}}\right\vert V_{\text{fin}})$ to denote the required set of
$\mathcal{N}^{\prime}$\ microstates.\ The $\mathcal{N}-\mathcal{N}^{\prime}$
remaining microstates in $\boldsymbol{\Gamma}_{\text{diff}}$ are superfluous.
We can similarly extend the above discussion to $V_{\text{fin}}=3V_{\text{in}%
}$, and then to $2V_{\text{in}}<V_{\text{fin}}<2V_{\text{in}}$, and so on.

With the above understanding of $\boldsymbol{\Gamma}_{\text{diff}}(\left.
V_{\text{in}}\right\vert V_{\text{fin}})$, there will be no confusion to
simplify the notation and use $\boldsymbol{\Gamma}_{\text{diff}}\ $for it.
With this understanding, $\boldsymbol{\Gamma}_{\text{fin}}$ as the union of
$\boldsymbol{\Gamma}_{\text{in}}$\ and $\boldsymbol{\Gamma}_{\text{diff}}$
only contains the required $\mathcal{N}+\mathcal{N}^{\prime}$ microstates.
This will be understood below.

Let us pick two microstates $\delta\mathbf{z}(V_{\text{in}})\in
\boldsymbol{\Gamma}_{\text{in}}$ and $\delta\mathbf{\zeta}(V_{\text{in}}%
)\in\boldsymbol{\Gamma}_{\text{diff}}$. Their image microstates $\delta
\mathbf{z}^{\prime}(V_{\text{fin}})\in\boldsymbol{\Gamma}_{\text{in}}^{\prime
}$ and $\delta\mathbf{\zeta}^{\prime}(V_{\text{fin}})\in\boldsymbol{\Gamma
}_{\text{diff}}^{\prime}$ are obtained by the deterministic evolution along
$\gamma_{\text{in}}=\gamma(\delta\mathbf{z})$ and $\gamma_{\text{fin}}%
=\gamma(\delta\mathbf{\zeta})$, respectively. The corresponding microworks are
given by%
\end{subequations}
\begin{equation}%
\begin{array}
[c]{c}%
\Delta W(\delta\mathbf{z})=-(E(\delta\mathbf{z}^{\prime})-E(\delta
\mathbf{z})),\\
\Delta W(\delta\mathbf{\zeta})=-(E(\delta\mathbf{\zeta}^{\prime}%
)-E(\delta\mathbf{\zeta})).
\end{array}
\label{MicroWork-PhaseSpace}%
\end{equation}
However, the situation is still not fully resolved as $\delta\mathbf{\zeta
(Z}_{\text{in}}^{E})\in\boldsymbol{\Gamma}_{\text{diff}}$ does not represent a
physical initial microstate in the left chamber; instead, it refers to a
microstate in the right chamber so its probability vanishes:
\[
p_{\text{in}}(\delta\mathbf{\zeta})=0,\delta\mathbf{\zeta}\in
\boldsymbol{\Gamma}_{\text{diff}}.
\]

We recall that we do not need $p_{k}$\ to determine $\Delta W_{k}$\ so it does
not matter if $p_{\text{in}}(\delta\mathbf{\zeta})=0$. During expansion,
$p(\delta\mathbf{\zeta})$ at $t>0$ is not going to remain zero. Therefore, we
\emph{formally assume} that the initial probability distribution
$p_{\text{in}}(\delta\mathbf{\zeta})$ is \emph{infinitesimally small} for
$\delta\mathbf{\zeta}$\ by shifting its initial energy by a very large
positive contribution
\begin{equation}
E(\delta\mathbf{\zeta})\rightarrow E(\delta\mathbf{\zeta})+e(\delta
\mathbf{\zeta})/\varepsilon>0,\delta\mathbf{\zeta}\in\boldsymbol{\Gamma
}_{\text{diff}}\text{ at }t=0\label{Energy-Infinite}%
\end{equation}
using an infinitesimal quantity $\varepsilon>0$. At the end of the
calculation, the limit $\varepsilon\rightarrow0^{+}$ will be taken to ensure
$p_{\text{in}}(\delta\mathbf{\zeta})\overset{\varepsilon\rightarrow0^{+}%
}{\rightarrow}0$. This trick of energy shift transforms $\mathcal{W}$ as
\begin{align}
\mathcal{W}(\left.  \dot{\beta},\varepsilon\right\vert \Delta\mathsf{W}) &
\rightarrow%
{\textstyle\sum\limits_{\delta\mathbf{z}\in\boldsymbol{\Gamma}_{\text{in}}}}
p(\gamma_{\text{in}})e^{\dot{\beta}\Delta W(\delta\mathbf{z})}\nonumber\\
&  +%
{\textstyle\sum\limits_{\delta\mathbf{\zeta}\in\boldsymbol{\Gamma
}_{\text{diff}}}}
p(\gamma_{\text{fin}})e^{\dot{\beta}\Delta W(\delta\mathbf{\zeta}%
)};\label{Shifted-W}%
\end{align}
however, the energy shift leaves $\Delta W(\delta\mathbf{\zeta})$\ invariant.
Thus, we see that letting $\varepsilon\rightarrow0^{+}$ makes the second term
vanish so $\mathcal{W}(\left.  \dot{\beta},\varepsilon\right\vert
\Delta\mathsf{W})\overset{\varepsilon\rightarrow0^{+}}{\rightarrow}%
\mathcal{W}(\left.  \dot{\beta}\right\vert \Delta\mathsf{W})$. As
$\mathcal{W}$ and all the moments remain invariant, the \emph{trick} (energy
shift followed by $\varepsilon\rightarrow0^{+}$) does not affect
thermodynamics in anyway. In particular, it leaves the first law unaffected.
We will verify this below by direct manipulation.

It is clear that the two sums in Eq. (\ref{Shifted-W}) can be expressed as a
single sum over \emph{all} microstates $\delta\mathbf{\bar{z}}\in
\boldsymbol{\Gamma}_{\text{fin}}$, which refers to microstates in
$\boldsymbol{\Gamma}_{\text{fin}}$. For example, the shifted initial partition
function also remains invariant under the trick:%
\begin{equation}
Z_{\text{in}}(\beta_{0},V_{\text{in}},\varepsilon)\doteq%
{\textstyle\sum\limits_{\delta\mathbf{\bar{z}}\in\boldsymbol{\Gamma
}_{\text{fin}}}}
e^{-\beta_{0}E(\delta\mathbf{\bar{z}})}\overset{\varepsilon\rightarrow0^{+}%
}{\rightarrow}Z_{\text{in}}(\beta_{0},V_{\text{in}}). \label{Limit-PF}%
\end{equation}
This is consistent with Theorem \ref{Theorem-p_k-in-independence}. Thus, we
can focus on $\boldsymbol{\Gamma}_{\text{fin}}$ from the start instead of
$\boldsymbol{\Gamma}_{\text{in}}$. This allows us to basically use the
identity mapping $\mathcal{T}$ between initial microstates $\delta
\mathbf{\bar{z}}\in\boldsymbol{\Gamma}_{\text{fin}}$ and final microstates
$\delta\mathbf{\bar{z}}^{\prime}\in\boldsymbol{\Gamma}_{\text{fin}}^{\prime
}\ $using Hamiltonian trajectories $\gamma$ as $(V_{\text{in}},\xi_{\text{in}%
})\rightarrow(V_{\text{fin}},\xi_{\text{fin}}=0)$. In the following, we will
use $\left\{  \delta\mathbf{\bar{z}}\right\}  $ and $\left\{  \delta
\mathbf{\bar{z}}^{\prime}\right\}  $ to show these sets. We will use $\left\{
\overline{\gamma}(t)\right\}  $ to denote the ensemble of Hamiltonian
trajectories at any time $t$.

We can combine the two equations in Eq. (\ref{MicroWork-PhaseSpace}) in a
single equation $\Delta W(\delta\mathbf{\bar{z}})=-(E(\delta\mathbf{\bar{z}%
}^{\prime})-E(\delta\mathbf{\bar{z}}))$, with $E(\delta\overline{\mathbf{z}})$
playing the role of $E_{k}$ along $\overline{\gamma}$. Let us consider
$\overline{\mathcal{W}}$ for $\mathcal{P}_{0}$
\[
\overline{\mathcal{W}}(\left.  \beta_{0}\right\vert \left\{  \Delta
\mathsf{W}\right\}  )\doteq\lim_{\varepsilon\rightarrow0^{+}}%
{\textstyle\sum\nolimits_{\delta\mathbf{\bar{z}}\in\boldsymbol{\Gamma
}_{\text{fin}}}}
p_{\text{eq}}(\overline{\gamma})e^{\beta_{0}\Delta W(\delta\mathbf{\bar{z}}%
)}.
\]
Recalling Eqs. (\ref{trajectory prob-cond}) and
(\ref{probability-dist-initial}), the summand becomes $p_{\text{eq}}(\left.
\delta\mathbf{\bar{z}}^{\prime}\right\vert \delta\mathbf{\bar{z}}%
)e^{-\beta_{0}E(\delta\mathbf{\bar{z}}^{\prime})}/Z_{\text{in}}(\beta
_{0},V_{\text{fin}},\varepsilon)$, in which $e^{\beta_{0}E(\delta
\mathbf{\bar{z}})}$ from $p_{\text{in,eq}}(\delta\mathbf{\bar{z}})$ exactly
cancels with $e^{\beta_{0}E(\delta\mathbf{\bar{z}})}$ coming from $\Delta
W(\delta\mathbf{\bar{z}})$. As the conditional probability $p(\left.
\delta\mathbf{\bar{z}}^{\prime}\right\vert \delta\mathbf{\bar{z}})$ is not
affected by $p_{\text{in}}(\delta\mathbf{\bar{z}})$, taking the limit is
trivial as it affects only $Z_{\text{in}}(\beta_{0},V_{\text{fin}}%
,\varepsilon)$ in accordance with Eq. (\ref{Limit-PF}) so the above summand
converges to $p(\left.  \delta\mathbf{\bar{z}}^{\prime}\right\vert
\delta\mathbf{\bar{z}})e^{-\beta_{0}E(\delta\mathbf{\bar{z}}^{\prime}%
)}/Z_{\text{in}}(\beta_{0},V_{\text{in}})$. Thus,%
\begin{equation}
\overline{\mathcal{W}}(\left.  \beta_{0}\right\vert \Delta\mathsf{W})=%
{\textstyle\sum\limits_{\delta\mathbf{\bar{z}}\in\boldsymbol{\Gamma
}_{\text{fin}}}}
\frac{p_{\text{eq}}(\left.  \delta\mathbf{\bar{z}}^{\prime}\right\vert
\delta\mathbf{\bar{z}})e^{-\beta_{0}E(\delta\mathbf{\bar{z}}^{\prime})}%
}{Z_{\text{in}}(\beta_{0},V_{\text{in}})}. \label{MFF-ClassicalGas}%
\end{equation}

Before proceeding further, we wish to confirm that the inclusion of
microstates in $\Gamma_{\text{diff}}$ does not violate the first law in MNEQT
by focusing on $\left\langle \Delta\mathsf{W}\right\rangle $. We consider a
microstate $\delta\mathbf{\bar{z}}(t)$ along an infinitesimal segment
$d\gamma$ on $\gamma$ between $t>0$ and $t+dt$. Over this segment,
$\mathbf{Z}^{E}(t)=(V,\xi)\rightarrow\mathbf{Z}^{E}(t+dt)=(V+dV,\xi+d\xi)$.
Thus,%
\[
dW(t)=-%
{\textstyle\sum\limits_{\left\{  \delta\mathbf{\bar{z}}(t)\right\}  }}
p(\delta\mathbf{\bar{z}}(t))[E(\mathbf{Z}^{E}(t+dt))-E(\mathbf{Z}^{E}(t))].
\]
Introducing $dp(t)\doteq$ $p(\delta\mathbf{\bar{z}}(t+dt))-p(\delta
\mathbf{\bar{z}}(t))$, and $dE(t)=E(t+dt)-E(t)$, we have
\begin{align*}
dW(t)  &  =-dE(t)+%
{\textstyle\sum\limits_{\left\{  \delta\mathbf{\bar{z}}(t)\right\}  }}
dp(t)E(\mathbf{Z}^{E}(t+dt))\\
&  \simeq-dE(t)+dQ(t),
\end{align*}
where we have neglected the second-order term $dp(t)dE(\mathbf{Z}^{E}(t))$ as
is a common practice and used Eq. (\ref{Macro-Heats}) to identify $dQ$. Thus,
for any $t>0$, we have satisfied Eq. (\ref{FirstLaw-SI}). To consider $t=0$,
we need to take the limit $\varepsilon\rightarrow0^{+}$, which limits the sum
in $dW$ above to $\delta\mathbf{\bar{z}}\in\boldsymbol{\Gamma}_{\text{in}}$.
This ensures that Eq. (\ref{FirstLaw-SI}) remains satisfied at $t=0$ with our
modification in Eq. (\ref{Energy-Infinite}). Thus, we have established that
our trick of using $\forall\delta\mathbf{\bar{z}}\in\boldsymbol{\Gamma
}_{\text{fin}}$ is consistent with the MNEQT. It follows then that the
discussion here in the $\mu$NEQT will finally result in Eq.
(\ref{Entropy-Change-Eq}) once we realize that $dE=0$ so that $dW=TdS$ as in
Eq. (\ref{dS-dW-Free Expansion}).

We now consider $\widehat{\mathcal{W}}(\left.  \beta_{0}\right\vert
\Delta\mathsf{W})$, which can be exactly evaluated. For this, we set
$p_{\text{eq}}(\left.  \delta\mathbf{\bar{z}}^{\prime}\right\vert
\delta\mathbf{\bar{z}})=1$ above and use the $1$-to-$1$ mapping $\mathcal{T}%
:\delta\mathbf{\bar{z}\rightarrow}\delta\mathbf{\bar{z}}^{\prime}$ to replace
the sum over $\delta\mathbf{\bar{z}}$ to a sum over $\delta\mathbf{\bar{z}%
}^{\prime}$ to obtain
\begin{equation}
\widehat{\mathcal{W}}(\left.  \beta_{0}\right\vert \Delta\mathsf{W})=%
{\textstyle\sum\limits_{\delta\mathbf{\bar{z}}^{\prime}\in\boldsymbol{\Gamma
}_{\text{fin}}^{\prime}}}
\frac{e^{-\beta_{0}E(\delta\mathbf{\bar{z}}^{\prime})}}{Z_{\text{in}}%
(\beta_{0},V_{\text{in}})}=\frac{Z_{\text{fin}}(\beta_{0},V_{\text{fin}}%
)}{Z_{\text{in}}(\beta_{0},V_{\text{in}})}, \label{AverageExpWork3-1}%
\end{equation}
which is precisely Eq. (\ref{GujratiRelation0}). Notice that the initial EQ
macrostate in proving the above relation corresponds to the one with all the
particles confined in the left chamber; see, however, \cite{KestinNote}.

Above, we have considered the case of free expansion. The same trick will also
work if the expansion is gradual and not abrupt. The only difference will be
that $V(t)$ will not be an internal variable as it is controlled externally.
We still will need the trick of inserting "missing" microstates as above. By
interchanging the role of the initial and final phase spaces above, we can
also use the trick to investigate the case of contraction. The only difference
is that in the last two cases, the system is not isolated so we must make a
distinction between $dW$ (or $\Delta W$) and $d_{\text{i}}W$ (or
$\Delta_{\text{i}}W$).

If it happens that the phase space volume $\left\vert \Gamma\right\vert $
continues to change during a process but $\left\vert \Gamma_{\text{in}%
}\right\vert =\left\vert \Gamma_{\text{fin}}\right\vert $, such as a cyclic
process, we can treat it as a combination of expansion and contractions
processes.\ To see this, we look for the time $t_{\text{m}}$ when $\left\vert
\Gamma\right\vert $ is maximum (minimum). Then, we are dealing with expansion
(contraction) over $t>t_{\text{m}}$, and contraction (expansion) over
$t>t_{\text{m}}$. The same approach of a combination of expansion/contraction
can be taken when $\left\vert \Gamma\right\vert $ does not change
monotonically during a process, whether $\left\vert \Gamma_{\text{in}%
}\right\vert =\left\vert \Gamma_{\text{fin}}\right\vert $ or not.

\section{Discussion and Conclusion\label{Sec-Conclusion}}

The current investigation was motivated by a desire to understand the
following two very important aspects of NEQ thermodynamics at the microstate level

\begin{enumerate}
\item how to use Hamiltonian trajectories to describe phase space volume
changes in a process to construct a microstate NEQ thermodynamics of a system,
\emph{interacting or not}, and

\item the nature of microworks that give rise to the dissipated work in a
\emph{noninteracting}, \textit{i.e.,} an isolated system for which no exchange
of macroheat and macrowork is allowed ($\Delta_{\text{e}}Q=0,\Delta_{\text{e}%
}W=0$),
\end{enumerate}

as these issues have not been addressed in the literature but lie at the heart
of the many common NEQ processes including free expansion taught to
undergraduates in macroscopic thermodynamics. Recently, we have developed a
macroscopic and microscopic SI NEQ\ thermodynamics (the MNEQT and the $\mu
$NEQT) that directly include the macroforce imbalance $\mathbf{F}_{\text{t}}$
or the microforce imbalance $\mathbf{F}_{\text{t,}k}^{\text{w}}$, an important
concept introduced recently by us, to ensure describing an interacting and a
noninteracting system within the same framework. As discussed in Sec.
\ref{Sec-Prolog}, no other current NEQ thermodynamics directly captures the
microforce imbalance. In its absence ($\mathbf{F}_{\text{t,}k}^{\text{w}%
}=0,\forall k$), a situation common in various MFTs noted in Sec.
\ref{Sec-Prolog} and discussed in Sec. \ref{Sec-WFT}, there cannot be any work
irreversibility so we cannot overemphasize its importance for any NEQ processes.

The phase space $\Gamma$ of a classical system with finite parameter
$\mathbf{Z}$ has a finite volume $\left\vert \Gamma\right\vert $. In
particular, this requires $\left\{  E_{k}\right\}  $ to have only
\emph{finite} number of values. This is the case for all numerical simulations
so our approach here provides a useful approach to carry out numerical
simulation of a finite system.

Let us consider the following two cases that can arise.

1. If the volume $\left\vert \Gamma\right\vert $ does not change at all during
a process ($\left\vert \Gamma_{\text{in}}\right\vert =\left\vert
\Gamma_{\text{fin}}\right\vert $), the finite number of microstates
$\mathcal{N\doteq}\left\vert \mathfrak{m}\right\vert $ also does not change,
and one can always use the identity property of the Hamiltonian trajectory in
Eq. (\ref{Temporal mapping}) to follow the temporal evolution of each
microstate $\mathfrak{m}_{k}$ along $\gamma_{k}$.

2. When the volume $\left\vert \Gamma\right\vert $ changes monotonically
during a process, which is a common situation such as during expansion or
contraction of a classical system, then $\mathcal{N}_{\text{in}}%
\mathcal{\doteq}\left\vert \mathfrak{m}_{\text{in}}\right\vert \neq
\mathcal{N}_{\text{fin}}\mathcal{\doteq}\left\vert \mathfrak{m}_{\text{fin}%
}\right\vert $. Then these microstates cannot be connected $1$-to-$1$ by
Hamiltonian trajectories. The excess microstates $\mathcal{N}_{\text{diff}%
}\doteq\left\vert \mathfrak{m}_{\text{diff}}\right\vert \ $between $\left\{
\mathfrak{m}_{\text{in}}\right\}  $ and $\left\{  \mathfrak{m}_{\text{fin}%
}\right\}  $ will have no Hamiltonian trajectories\ associated with them as
seen in Fig. \ref{Fig_PhasePoint-Evolution} for the case of expansion. This is
when our trick in Sec. \ref{Sec-ClassicalExpansion} will be useful. We
introduce excess microstates in the phase space with the smaller number of
microstates but with extremely high positive microenergy shift to ensure that
$\left\vert \mathfrak{m}\right\vert $ is the same in the initial and final
macrostates. This trick is crucial and enables us to use only Hamiltonian
trajectories to connect the microstates in $\Gamma_{\text{in}}$ and
$\Gamma_{\text{fin}}$ in a $1$-to-$1$ manner. Eventually, the microenergy of
the missing microstates is taken to diverge to $+\infty$ to ensure their
probabilities vanish in the phase space of smaller volume. However, if
$\left\vert \Gamma\right\vert $ changes nonmonotonically during the process
even if $\left\vert \Gamma_{\text{in}}\right\vert =\left\vert \Gamma
_{\text{fin}}\right\vert $, then, as discussed in Sec.
\ref{Sec-ClassicalExpansion}, we are dealing with a combination of expansion
and contraction so it belongs to this case.

The above argument suggests that when phase space volume $\left\vert
\Gamma\right\vert $ is not finite, the Hamiltonian trajectories will not work,
since there is no way to argue that $\mathcal{N}_{\text{in}}\neq
\mathcal{N}_{\text{fin}}$ when both are infinitely large. This is not correct.
Let us assume that $\left\vert \mathfrak{m}\right\vert $ is infinitely large
but the set $\left\{  \mathfrak{m}_{k}\right\}  $ is denumerable as is the
case with the quantum expansion in Sec. \ref{Sec-QuantumExpansion}. We still
have $1$-to-$1$ Hamiltonian trajectories connecting $\mathfrak{m}%
_{k\text{,in}}\ $with $\mathfrak{m}_{k\text{,fin}}$,$\forall k\in\mathbb{N}$,
where $\mathbb{N}$ is the set of natural numbers, \textit{i.e.} \{1,2,3,\ldots
\}. It is this property that allows us to introduce the $\mu$NEQT in Sec.
\ref{Sec-muNEQT} and the MGF in Sec. \ref{Sec-MGF0}. This is true despite the
sets $\left\{  \mathfrak{m}_{k\text{,in}}\right\}  $ and $\left\{
\mathfrak{m}_{k\text{,fin}}\right\}  $\ having "equal" but infinite large
number of microstates so we can say that $\left\vert \mathfrak{m}_{\text{in}%
}\right\vert \equiv\left\vert \mathfrak{m}_{\text{fin}}\right\vert $.

In the classical expansion/contraction, the equality of the numbers of
microstates is clearly not valid despite the fact that they both have the same
cardinality as that of $\mathbb{N}$. What is important for the cardinality
consideration is the idea of association and not the unique $1$-to-$1$
mapping. This is the same when we compare $\mathbb{N}$ with the set
$\mathbb{N}_{\text{o}}$ of odd natural numbers, \textit{i.e.} \{1,3,5,\ldots\}
or the set $\mathbb{N}_{\text{e}}$ of even natural numbers, \textit{i.e}.
\{2,4,6,\ldots\}. All these sets have the same cardinality, but $\mathbb{N}$
has twice as many members as either of the sets $\mathbb{N}_{\text{o}}$ or
$\mathbb{N}_{\text{e}}$. It is the number of members that is important in the
mapping in Eq. (\ref{Temporal mapping}).

To understand it, we proceed as follows. Think of $\mathbb{N}$ as an example
of the set of microstates in the interior of the dotted ellipse $\Gamma
_{\text{fin}}$, $\mathbb{N}_{\text{o}}$ as an example of the set of
microstates in the interior of the solid ellipse $\Gamma_{\text{in}}$, and
$\mathbb{N}_{\text{e}}$ as an example of the set of microstates in the
difference $\boldsymbol{\Gamma}_{\text{diff}}$ of the two ellipses
$\Gamma_{\text{fin}}$ and $\Gamma_{\text{in}}$. We can similarly introduce
$\mathbb{N}^{\prime}$ for the interior of the solid ellipse $\Gamma
_{\text{fin}}^{\prime}$, $\mathbb{N}_{\text{o}}^{\prime}$ for the interior of
the broken ellipse $\Gamma_{\text{in}}^{\prime}$, and $\mathbb{N}_{\text{e}%
}^{\prime}$ for the difference $\boldsymbol{\Gamma}_{\text{diff}}^{\prime}$.
Under the $1$-to-$1$ mapping, $\mathbb{N}_{\text{o}}$ goes to $\mathbb{N}%
_{\text{o}}^{\prime}$, with no mapping leading to $\mathbb{N}_{\text{e}%
}^{\prime}$. This clearly shows that the concept of cardinality is irrelevant
as noted in Sec. \ref{Sec-VoluemChange}. Now we can use our trick introducing
excess microstates to ensure $\mathcal{N}_{\text{in}}=\mathcal{N}_{\text{fin}%
}$. Using the MGF $\mathcal{W}$, we have also established that the trick does
not destroy the consistency of the $\mu$NEQT with the MNEQT so the trick is
thermodynamically consistent. It is clear that the trick is very general and
will work in all cases of \emph{missing} microstates.

The most important signature of a NEQ process is the existence of
(generalized) thermodynamic force $\mathbf{F}_{\text{t}}$, see Eq.
(\ref{Generalized-Thermodynamic-Force}), which drives the system towards EQ.
It is common to study a NEQ isothermal process ($T=T_{0}\Rightarrow
F_{\text{t}}^{\text{h}}=0$) so we must look for a nonvanishing $\mathbf{F}%
_{\text{t}}^{\text{w}}$ in this case, regardless of whether the system is
isolated or interacting. In the former case, $\mathbf{F}_{\text{t}}^{\text{w}%
}$ reduces to $\mathbf{A}$ as seen in Eq. (\ref{EntropyDiff-Isolated}), and
contains $\mathbf{A}$ in the latter case. The presence of $\mathbf{F}%
_{\text{t}}^{\text{w}}$ in a NEQ system, including the free expansion that we
are interested in, is necessary along with the presence of internal variables.
Both of these features are absent in MFTs so they provide no guide in our
investigation as discussed in Sec. \ref{Sec-Prolog}. We are therefore left to
use our recently developed NEQ thermodynamics, which contains both features.
For the sake of continuity, we have provided a brief review of the MNEQT and
the $\mu$NEQT in Secs. \ref{Sec-MNEQT} and \ref{Sec-muNEQT}. As $\mathbf{F}%
_{\text{t,}k}^{\text{w}}$ is ubiquitous (it is present even in EQ), it is
clear that its absence in the MFTs and their inability to treat free expansion
make them very different from the $\mu$NEQT.

We now discuss some of the important results of our analysis .

The SI microwork $\Delta W_{k}$ along $\gamma_{k}$ is shown to be given by the
negative of the microenergy change $\Delta E_{k}$ so it is unaffected by a
constant shift proportional to $\varepsilon^{-1}$ in $E_{k}$. Moreover, the
addition of missing microstates does not change the MGF $\mathcal{W}(\left.
\beta\right\vert \Delta\mathsf{W})$. This function is important to justify
that the trick leaves our approach consistent with thermodynamics. The
presence of internal variables makes the MNEQT and $\mu$NEQT perfectly suited
to study any system, interacting or isolated. The success of our theory to
study the behavior of a Brownian particle \cite{Gujrati-LangevinEq} also shows
that they can also be used to study small systems.

We have shown that the best way to understand the origin of dissipated work is
to focus on $dW$ and $dQ$, with $d_{\text{i}}W=d_{\text{i}}Q\geq0$, as is done
in the MNEQT and $\mu$NEQT. In an isolated system, $dW=d_{\text{i}%
}W=dQ=d_{\text{i}}Q$. Therefore, only $d_{\text{i}}W_{k}$ resulting from
$\mathbf{F}_{\text{t,}k}^{\text{w}}$\ needs to be determined in the $\mu$NEQT,
which is easier to do with the use of the Hamiltonian trajectories. The
presence of $\mathbf{F}_{\text{t,}k}^{\text{w}}$ and internal variables for
isolated systems finally explains the source of irreversibility ($d_{\text{i}%
}S\geq0$). In their absence, the isolated system will be in EQ and will show
no dissipation.

As we have already applied the $\mu$NEQT to small-scale systems (Brownian
particles) in Ref. \cite{Gujrati-LangevinEq}, it is possible to extend the
expansion/contraction of the gas investigated here to study Maxwell's demon
and the problem of Landauer's eraser. We hope to come to these problems in future.

Finally, we mention some of the important predictions of our approach. We
restrict ourselves to free expansion here. While in general, $dW_{k}%
=d_{\text{i}}W_{k}$ can have any sign even though $d_{\text{i}}W\geq0$, we
must have $dW_{k}\geq0$ in free expansion in accordance with Corollary
\ref{Corollary-SpontaneousEnergyDecrease}, which is a general result for any
number of internal variables. Therefore, this general prediction can always be
used to validate an experiment or computation that is performed to obtain the
microwork distribution $\left\{  \Delta W_{k}\right\}  $. One should not see
any negative $\Delta W_{k}$. By determining $\Delta W$ and comparing it with
$\Delta W_{\text{isoth}}=-\Delta\overline{F}=T_{0}\Delta S$, see Eqs.
(\ref{Iso-Delta_W0}, \ref{Iso-Delta_W_i}), we can determine whether the
process is isothermal or not; see Eq. (\ref{Isolated-Dissipated-Macrowork}).

\begin{quote}
Declarations of interest: none
\end{quote}

\end{document}